\documentclass[1p,fleqn,authoryear,preprint]{elsarticle}
\usepackage{amsmath,amsfonts,amsthm,amssymb}
\usepackage[mathscr]{eucal}
\usepackage{mathtools}
\usepackage{eurosym}
\usepackage{graphicx}
\usepackage{ifpdf}
\usepackage{array}

\usepackage{ifpdf}

\ifpdf
\usepackage[bookmarks=false,
pdfstartview=FitH,linkbordercolor={0.5 1 1},
citebordercolor={0.5 1 0.5},unicode,
pagebackref,hyperindex]{hyperref}
\else
\usepackage{breakurl}                          
\fi

\newtheorem{lemma}{Lemma}
\newtheorem{remark}{Remark}

\newlength{\gnat}
\newcolumntype{R}{>{\scriptsize\raggedleft\arraybackslash$}p{\gnat}<{\hspace{0.33\gnat}$}}

\newcommand{\bse}{\boldsymbol{e}}
\newcommand{\bsf}{\boldsymbol{f}}
\newcommand{\bsx}{\boldsymbol{x}}
\newcommand{\bsy}{\boldsymbol{y}}
\newcommand{\bsr}{\boldsymbol{r}}
\newcommand{\bss}{\boldsymbol{s}}

\allowdisplaybreaks[1]

\begin{document}
\begin{frontmatter}
\title{Double Exponential Instability of Triangular Arbitrage Systems}
\author[uc]{Rod Cross}
\ead{rod.cross@strath.ac.uk}

\author[ittp]{Victor Kozyakin\corref{t1}}
\ead{kozyakin@iitp.ru}

\address[uc]{Department of Economics,
University of Strathclyde,\\ Sir William Duncan Building, 130
Rottenrow, Glasgow, G4 OGE, Scotland}

\address[ittp]{Institute for Information Transmission Problems, Russian
Academy of Sciences,\\
Bolshoj Karetny lane 19, Moscow 127994 GSP-4, Russia}

\begin{abstract}
If financial markets displayed the informational efficiency postulated in
the efficient markets hypothesis (EMH), arbitrage operations would be
self-extinguishing. The present paper considers arbitrage sequences in
foreign exchange (FX) markets, in which trading platforms and information
are fragmented. In \cite{KozCalPok:ArXiv10} and \cite{CrossKPP:MECA12} it
was shown that sequences of triangular arbitrage operations in FX markets
containing $4$ currencies and trader-arbitrageurs tend to display
periodicity or grow exponentially rather than being self-extinguishing.
This paper extends the analysis to $5$ or higher-order currency worlds.
The key findings are that in a $5$-currency world arbitrage sequences may
also follow an exponential law as well as display periodicity, but that in
higher-order currency worlds a double exponential law may additionally
apply. There is an ``inheritance of instability'' in the higher-order
currency worlds. Profitable arbitrage operations are thus endemic rather
that displaying the self-extinguishing properties implied by the EMH.
\end{abstract}

\begin{keyword}
Limits to arbitrage\sep Recurrent sequences\sep Matrix products\sep
Asyn\-chron\-ous systems

\medskip
\noindent\textit{MSC 2000:} 91B26, 91B54, 91B64, 15A60

\medskip
\noindent\textit{JEL Classification}: C60, F31, D82
\end{keyword}

\end{frontmatter}


\section{Introduction}\label{S-intro}

Arbitrage operations are profitable if a good or asset can be bought for a
lower price than that for which it can be sold. Such operations are distinct
from speculation in that there is little or no capital or risk exposure.
Arbitrage profits are made by trading on prices that are already posted,
rather than by trading on speculative guesses about what prices posted in
the future might be. The \textsl{New Palgrave Dictionary} defines arbitrage
as ``an investment strategy that guarantees a positive payoff in some
contingency with no possibility of negative payoff and with no net
investment'', \cite{dybvig2008}.

In the economics and finance literature the exploitation of profitable
arbitrage opportunities is postulated to eliminate price discrepancies
between goods and assets that are in some sense ``identical.'' This
postulate, that arbitrage operations are self-extinguishing, implies the
``law of one price'', this terminology having been coined in the purchasing
power parity explanation of foreign exchange rates, \cite{Cassel1916}. A
related no-arbitrage condition is that of covered interest parity (CIP),
whereby the ratio of forward to spot exchange rates for currency pairs is
equal to the ratio of the interest rates on comparable assets denominated in
the two currencies over the forward period in question
\cite[p.~13]{keynes1923}. No-arbitrage condition form the bedrock of
mainstream theory in finance, as embodied in the Modigliani--Miller theorem
regarding corporate capital structure, the Black--Scholes model of option
pricing and in the arbitrage pricing model of asset prices, \cite{ross1978}.

The present paper deals with arbitrage processes in foreign exchange (FX)
markets. The Bank for International Settlements \cite[p.~5]{BIS11} provides
a useful classification of arbitrage strategies in such markets.
\emph{Classical arbitrage} aims to exploit discrepancies between actual
exchange rates and no-arbitrage conditions, such as that the exchange rates
for currency pairs are consistent with their cross-exchange rates.
\emph{Latency arbitrage} aims to exploit time lags between trades being
initiated and FX price quotes being revised. \emph{Liquidity imbalance
arbitrage} attempts to take advantage of order book imbalances between
different trading platforms. \emph{Complex event arbitrage} is geared
towards properties of FX rates, such as mean-reversion and momentum, or
towards systematic patterns in the way FX response to new information being
released. This paper focuses on triangular arbitrage operations, arising
from discrepancies between the exchange rates for currency pairs and the
cross-exchange rates for the currencies involved. To simplify the analysis
we ignore the difference between FX bid and ask prices and do not take
account of the cost of executing FX transactions. For evidence that there
are arbitrage profits to be made, net of bid-ask spreads and transactions
costs, see \cite{MTY} and \cite{ARS08}.

If the efficient market hypothesis (EMH), formulated by \cite{Sam65} and
\cite{Fama70}, applied to FX markets, exchange rates would reflect all the
information relevant to their determination and there would be no profits to
be had from FX arbitrage trades. This would involve an ``arbitrage paradox''
\citep{GS76}: if arbitrage opportunities did not exist, there would be no
incentive for FX market participants to monitor exchange rates, so
profitable arbitrage opportunities could well arise. There is also the ``no
trade'' problem. If FX markets were informationally efficient, as the EMH
postulates, you would expect to observe periods in which FX transactions
volumes dwindled when no new information was being generated. Instead FX
transactions volumes remain substantial, even during quiescent market
conditions.

The microstructure of FX markets is fragmented, implying that information
about profitable arbitrage opportunities is also likely to be fragmented.
The contrast is between the ``lit'' areas where information on some of the
electronic trading platforms is publicly available, and the ``dark pools'',
where information on FX transactions conducted directly between banks and
their end-user clients is initially private. BIS data for 2010 document this
FX market fragmentation. Electronic booking systems (EBS), such as the
London-based EBS and Reuters, took $18.8\%$ of global turnover; multi-bank
electronic communication networks (ECNs), such as the US-based Currenex,
Hotspot FX and FXall, accounted for $11.1\%$ of global turnover; and
single-bank ECNs took $11.4\%$ of global turnover. Inter-dealer trades, most
of which are executed electronically, constituted $18.5\%$ of turnover. The
``dark pools'' tend to be in the non-electronic execution segments of the
market, with customer-direct transactions between banks and end-user clients
taking $24.4\%$ of global turnover, and voice-broker trades accounting for
$15.9\%$ of global turnover \citep[p.~16]{BIS}.

There is also evidence of ``home bias'' in FX transactions. FX trades
initiated in the US and Canada tend to take place during North American
trading hours, those initiated in Japan and Australia tend to occur during
Asian trading hours, and so on \citep[Table~2]{Souza}. Estimation of
impulse-response functions for the way exchange rates respond to local order
flow data indicate that ``dealers operating both at the same time and in the
same geographic region as fundamentally driven customers have a natural
informational advantage'' \citep[pp.~23--24]{Souza}. In a similar vein,
\cite{CM02} found that Tokyo-based traders had an informational advantage
over foreign-based traders regarding the course of the Japanese yen exchange
rates.

In previous papers, \cite{KozCalPok:ArXiv10} and \cite{CrossKPP:MECA12}, the
implications of this fragmentation of information on FX markets were
analysed by a combinatorial analysis of the different possible sequences of
arbitrage operations. The ``home bias'' asymmetry in information is
represented by having FX trader-arbitrageurs initially know only the
exchange rates involving their own domestic currency. The question is then
whether triangular arbitrage transactions are profitable because the
cross-exchange rates are misaligned in relation to the exchange rates for
the currency pairs. The arbitrage sequences that will be pursued would
depend on which trader-arbitrageur first discovers the mis-alignment of the
cross-exchange rates. The arbitrage operations are triangular, so, for
example, if the US dollar trader were to discover that the euro--sterling
rate was out of line with exchange rates for the dollar--euro and
dollar--sterling currency pairs, an arbitrage profit could be made by
selling dollars for sterling and using the sterling to buy euros.

In the $3$-currency case arbitrage operations are reasonable
straightforward. The order in which trader-arbi\-trageurs discover
information about cross-exchange rate discrepancies makes a difference to
the arbitrage sequences that will be pursued, and to the resulting new
no-arbitrage ensemble of exchange rates that will emerge, but one arbitrage
transaction suffices to eliminate the arbitrage opportunity. The
$4$-currency case is significantly more complicated, there being
$^{4}\!P_{3} =24$ possible triangular arbitrage operations. In
\cite{KozCalPok:ArXiv10} and \cite{CrossKPP:MECA12} it was shown that in
this $4$-currency world arbitrage sequences tend to be periodic in nature or
display exponential behaviour, showing no tendency to approach a
no-arbitrage ensemble of exchange rates in which there are no profitable
opportunities. The intuitive explanation for this periodicity is that the
exploitation of one arbitrage opportunity has ``ripple effects,'' disturbing
the ensemble of exchange rates and so creating further active arbitrage
opportunities.

The task of the present paper is to extend the analysis of FX arbitrage
operations to a world of $d$-currencies, where $d$ is the number of
currencies and trade-arbitrageurs involved. One key finding is that in the
$(d\geq5)$-currency case the surprising result is that the arbitrage
sequences may follow a double exponential process as well as the periodicity
or exponential behaviour observed in the $4$-currency case. There is thus an
``inheritance of instability'' as we move to higher-order currency worlds.

The structure of the paper is as follows. In Section~\ref{S-problem} a
mathematical formulation of the problem is given. In Section~\ref{S-linprob}
it is shown that the arbitrage dynamics, initially specified in terms of
some nonlinear operations, can be described by asynchronous matrix products.
For the general case of an arbitrary number of trader-arbitrageurs the
properties of the resulting matrices are studied in
Section~\ref{S-asynchro}. The exposition of Sections~\ref{S-problem} and
\ref{S-linprob} follows the work of \cite{KozCalPok:ArXiv10}. In
Section~\ref{S-d4} basic results from \cite{KozCalPok:ArXiv10} and
\cite{CrossKPP:MECA12} about the arbitrage dynamics for an FX currency
market with $4$ traders are recalled. In particular, these results
demonstrate that the exchange rates in a foreign exchange currency market
with $4$ traders, under appropriate choice of the arbitrage sequences,
display periodic or exponential behaviour. In Section~\ref{S-d5} we
construct an example showing that in the case of $5$ or more arbitrage
traders the situation may be, in a sense, even worse --- in this case the
exchange rates may change not only periodically or grow exponentially, as in
the case of $4$ traders, but they may grow in accordance with a double
exponential law. In~\ref{A-d4} and \ref{A-d5} an explicit form for all the
matrices involved in the description of the arbitrage dynamics for the cases
$d=4$ and $d=5$, respectively, is presented.

\section{Statement of the Problem}\label{S-problem}
Consider a foreign exchange (FX) currency market involving $d$ currencies
involving the exchange rates $r_{ij}$ of the $i$-th currency to the $j$-th
currency, $j\neq i$. Naturally, the value of $r_{ji}$, the exchange rate of
the $j$-th currency to the $i$-th currency, is reciprocal to the value of
$r_{ij}$:
\[
r_{ji}=\frac{1}{r_{ij}}.
\]

The exchange rates vary with time depending on the state of the market and
on the relations between the exchange rates. For instance, if at some moment
the trader of the currency $i$ realises that exchange rate of the $i$-th
currency to the $j$-th currency, via the intermediate currency $k$, may
bring a profit, i.e.,
\[
r_{ik}\cdot r_{kj} > r_{ij},\quad i\neq j,~ k\neq i,j,
\]
then he can establish a new exchange rate for the currency
\[
r_{ij,\textrm{new}}=r_{ik}\cdot r_{kj}.
\]
So, in what follows, we will suppose that the exchange rates in our
arbitrage system are updated in accordance with the following law:
\begin{align*}
r_{ij,\textrm{new}}&=\max\left\{r_{ik}\cdot
r_{kj},~r_{ij,\textrm{old}}\right\},\\
r_{ji,\textrm{new}}&=\frac{1}{r_{ij,\textrm{new}}},
\end{align*}
and simultaneously the exchange rates for only one pair $(i,j)$ of
currencies may be updated.

By introducing the auxiliary quantities
\[
a_{ij}=\log r_{ij},\quad \forall i\neq j,
\]
it is possible to pass from the ``multiplicative'' statement of the problem
about arbitrage dynamics given above to the ``additive'' statement of the
problem, which in this case will look as follows. Given a skew-symmetric
$d\times d$ matrix $A=(a_{ij})$, for a triplet of pairwise distinct indices
$(i,j,k)$, $i\neq j$, $k\neq i,j$, the elements $a_{ij}$ and $a_{ji}$ are
updated in accordance with the following law:
\begin{align}\label{E-balance1}
    a_{ij,\textrm{new}}&=\max\left\{a_{ik}+a_{kj},~a_{ij,\textrm{old}}\right\},\\
\label{E-balance2}    a_{ji,\textrm{new}}&=-a_{ij,\textrm{new}}.
\end{align}
The triplet of indices $\omega=(i,j,k)$, $i\neq j$, $k\neq i,j$,
will be called  the \emph{arbitrage rule}.

\section{Linear Reformulation}\label{S-linprob}

Since by \eqref{E-balance1} and \eqref{E-balance2} the element $a_{ij}$ of
the matrix $A$ is updated ``simultaneously'' with the updating of the
symmetric element $a_{ji}$, then it is reasonable to speak about updating of
the pair $(a_{ij},a_{ji})$. Then, according to \eqref{E-balance1} and
\eqref{E-balance2}, the pair of elements $(a_{ij},a_{ji})$ can be updated by
one of two following scenarios. Either at first the element $a_{ij}$ is
updated in accordance with \eqref{E-balance1}, and then the element $a_{ji}$
is updated in accordance with \eqref{E-balance2}, or at first the element
$a_{ji}$ is updated by the formula
\[
a_{ji,\textrm{new}}=\max\left\{a_{jk}+a_{ki},~a_{ji,\textrm{old}}\right\},
\]
which, in view of the skew-symmetry of the matrix $A$, leads to the formula
\[
a_{ji,\textrm{new}}=\max\left\{-a_{kj}-a_{ik},~-a_{ij,\textrm{old}}\right\}=
-\min\left\{a_{ik}+a_{kj},~a_{ij,\textrm{old}}\right\},
\]
and then the element $a_{ij}$ is updated as follows
\[
a_{ij,\textrm{new}}=-a_{ji,\textrm{new}},
\]
which amounts to the final expression for $a_{ij,\textrm{new}}$:
\begin{equation}\label{E-balance3}
a_{ij,\textrm{new}}=\min\left\{a_{ik}+a_{kj},~a_{ij,\textrm{old}}\right\}.
\end{equation}

Formulae \eqref{E-balance1} and \eqref{E-balance3} show that in the process
of updating of the pair $(a_{ij},a_{ji})$ the element $a_{ij}$,
independently of the scenario of updating, either is not changed (and then
the matrix $A$ is not changed, too) or this element may change its value as
follows:
\begin{equation}\label{E-balance4}
a_{ij,\textrm{new}}=a_{ik}+a_{kj},
\end{equation}
whereas the symmetric element $a_{ji}$ changes its value in accordance with
\eqref{E-balance2} which is equivalent to
\begin{equation}\label{E-balance5}
    a_{ji,\textrm{new}}=a_{jk}+a_{ki}.
\end{equation}

So, the following lemma is proved.

\begin{lemma}\label{Lem-direct}
If the arbitrage law \eqref{E-balance1} and \eqref{E-balance2} holds
then any pair of elements $(a_{ij},a_{ji})$ either is not changed
or is changed by the \textbf{linear law} \eqref{E-balance4},
\eqref{E-balance5}.
\end{lemma}

Clearly, the converse assertion is also true.

\begin{lemma}\label{Rem-inverse}
If a pair of elements $(a_{ij},a_{ji})$ were to be updated by the law
\eqref{E-balance4}, \eqref{E-balance5}, i.e.,
\[
(a_{ij},a_{ji})\mapsto (a_{ij,\mathrm{new}},a_{ji,\mathrm{new}})
\]
then one of two scenarios of updating of the pair of elements
$(a_{ij},a_{ji})$ can be selected (at first updating $a_{ij}$ by formula
\eqref{E-balance1} and then adjusting $a_{ji}$ by the formula
\eqref{E-balance2}; or at first updating $a_{ji}$ by the formula
\eqref{E-balance1} and then adjusting $a_{ij}$ by the formula
\eqref{E-balance2}), under which the same pair of new elements
$(a_{ij,\mathrm{new}},a_{ji,\mathrm{new}})$ will be obtained. Namely,
\begin{itemize}
\item if $a_{ij,\mathrm{new}}>a_{ij}$ after updating the pair
    $(a_{ij},a_{ji})$ then, to interpret such an updating by the arbitrage
    law \eqref{E-balance1} and \eqref{E-balance2}, one should first apply
    formula \eqref{E-balance1} and then formula \eqref{E-balance2};

\item if $a_{ij,\mathrm{new}}<a_{ij}$ after updating the pair
    $(a_{ij},a_{ji})$ then, to interpret such an updating by the arbitrage
    law \eqref{E-balance1} and \eqref{E-balance2}, one should first apply
    formula \eqref{E-balance2} and then formula \eqref{E-balance1};

\item and last, if $a_{ij,\mathrm{new}}=a_{ij}$ then the order of updating
    the elements of the pair $(a_{ij},a_{ji})$, as well as whether such an
    updating took place at all, is inessential.
\end{itemize}
\end{lemma}

The previous lemmata indicate that under the max-statement of the problem,
as well as under its linear reformulation, it suffices to investigate only
the dynamics of the pairs $(a_{ij},a_{ji})$ of mutually symmetric elements
since such an analysis \emph{allows us to reproduce also the scenarios of
intra-pair updating of the elements $a_{ij}$ under the max-statement of the
problem}.

\section{Asynchronous Matrix Products}\label{S-asynchro}

Define in the space of all skew-symmetric $d\times d$ matrices $A=(a_{ij})$
a basis $\{\bse_{ij}\}$, $1\le i<j\le d$, by enumerating in some order the
upper off-diagonal elements of such matrices. For example, let us define the
element $\bse_{ij}$, $1\le i<j\le d$, of this basis as the skew-symmetric
matrix whose element $a_{ij}$ is equal to $1$ and the others vanish. Then
each matrix in this basis can be represented as a column-vector:
\begin{equation}\label{E-matcoord}
\bsx=\left\{a_{12}, a_{13},\ldots, a_{1d}, a_{23},
a_{24},\ldots\,a_{2d},\ldots,a_{d-1,d}\right\}^{T}\in\mathbb{R}^{d(d-1)/2},
\end{equation}
and computation of the new vector $\bsx_{\textrm{new}}$ by the old one
$\bsx_{\textrm{old}}$ will be defined by the expression
\begin{equation}\label{E-liniter}
\bsx_{\textrm{new}}=B_{\omega}\bsx_{\textrm{old}}.
\end{equation}
Here the matrix $B_{\omega}$ is determined by the arbitrage rule
$\omega=(i,j,k)$ applied to this step of updating and has the form typical
in the theory of asynchronous systems, see, e.g., \cite{AKKK:92:e}: all of
its rows except one coincide with the rows of the identity matrix while one
row contains exactly two non-zero elements specified by formula
\eqref{E-balance4}.

If we define on the space of vectors \eqref{E-matcoord} the usual Euclidean
inner-product $\langle\cdot,\cdot\rangle$, and then use the matrix
$B_{\omega}$ in \eqref{E-liniter} corresponding to the arbitrage rule
$\omega=(i,j,k)$, we will get the following representation\footnote{Recall
that always $i<j$ and $k\neq i,j$.}:
\begin{equation}\label{E-defB}
B_{\omega}\bsx=\begin{cases}
\bsx+\langle -\bse_{ij}-\bse_{ki}+\bse_{kj},\bsx\rangle
\bse_{ij}&\textrm{for}\quad k<i<j,\\
\bsx+\langle -\bse_{ij}+\bse_{ik}+\bse_{kj},\bsx\rangle
\bse_{ij}&\textrm{for}\quad i<k<j,\\
\bsx+\langle -\bse_{ij}+\bse_{ik}-\bse_{jk},\bsx\rangle
\bse_{ij}&\textrm{for}\quad i<j<k.
\end{cases}
\end{equation}
Since all the vectors \eqref{E-matcoord} are column-vectors, then the
relations \eqref{E-defB} can also be rewritten in the following matrix
form:
\begin{equation}\label{E-defBx}
B_{\omega}=\begin{cases}
I-\bse_{ij}(\bse_{ij}^{T}+\bse_{ki}^{T}-\bse_{kj}^{T})&\textrm{for}\quad
k<i<j,\\
I-\bse_{ij}(\bse_{ij}^{T}-\bse_{ik}^{T}-\bse_{kj}^{T})&\textrm{for}\quad
i<k<j,\\
I-\bse_{ij}(\bse_{ij}^{T}-\bse_{ik}^{T}+\bse_{jk}^{T})&\textrm{for}\quad
i<j<k,
\end{cases}
\end{equation}
where the upper index $T$ denotes transposition of a vector or a matrix.

Let us make some remarks resulting from formulae \eqref{E-defB} and
\eqref{E-defBx}.

\begin{remark}\label{Rem1}
Each matrix $B_{\omega}$ is a matrix with \textsl{integer entries}. This
means that, when considering convergence issues for the matrix products
involving such matrices, the convergence, provided that it takes place,
\textsl{will be always achieved by a finite number of steps}. \qed
\end{remark}

\begin{remark}\label{Rem2}
Each matrix $B_{\omega}$ is a \textsl{projector}. \qed
\end{remark}

\begin{remark}\label{Rem5}
Since the procedure of sequential updating of vectors $\boldsymbol{x}$
according to \eqref{E-liniter} linearly depends on the initial vectors, then
to analyse their convergence it suffices to consider only the initial
vectors taking the values of the standard basis vectors $\{\bse_{ij}\}$.
Then the whole problem becomes an integer-valued problem, i.e, a
\emph{combinatorial} one. \qed
\end{remark}

\begin{remark}\label{Rem3}
The procedure of constructing the matrices $B_{\omega}$ is very similar to
that of constructing the so-called ``mixtures'' of matrices under
investigation in asynchronous systems \cite{AKKK:92:e}. The difference is
that in \cite{AKKK:92:e} each coordinate of the state vector
$\bsx_{\textrm{new}}$ in \eqref{E-liniter} is updated by a single rule,
whereas in our case each coordinate of the state vector
$\bsx_{\textrm{new}}$ in \eqref{E-liniter} may be updated in accordance with
several rules --- see Section~\ref{S-d4} below. This is explained by the
fact that the set of all the arbitrage rules $(i,j,\cdot)$ having the same
first two indices $i,j$, in general, contains more than one element.

Investigation of convergence of the procedures \eqref{E-liniter} has very
much in common with investigation of the joint/gen\-eral\-ized spectral
radius of the family of all the matrices $B_{\omega}$, see, e.g.,
\cite{Jungers:09} and bibliography therein.

So, a great variety of results from the theory of asynchronous systems and
from the theory of the joint/generalized spectral radius might be helpful
here. \qed
\end{remark}

Clearly the arbitrage procedure \eqref{E-liniter} may have limiting fixed
states (with respect to all the matrices $B_{\omega}$) if and only if the
eigenspaces of the matrices $B_{\omega}$, corresponding to the eigenvalue
$1$, have nontrivial intersection $\mathbb{F}$. To find this subspace
$\mathbb{F}$ it is necessary to solve the following system of linear
equations:
\[
\bsx=B_{\omega}\bsx,\quad \forall~\omega=(i,j,k):~ 1\le i<j\le d,~ k\neq
i,j.
\]

\begin{lemma}\label{Lem-fixpoints}
The subspace $\mathbb{F}$ of common fixed points of all the matrices
$\{B_{\omega}\}$ consists of the column-vectors $\bsx=(x_{ij})^{T}$
satisfying
\[
x_{ij}=-x_{1i}+x_{1j},\qquad 2\le i<j\le d,
\]
where $x_{12},x_{13},\ldots,x_{1d}$ are free variables.
\end{lemma}

\begin{proof} By induction with respect to $d\ge 3$.
\end{proof}

\begin{remark}\label{Rem-fixmat}
By Lemma~\ref{Lem-fixpoints}, one can assert  that a skew-symmetric
matrix $A$ is unchangeable by application of any arbitrage rule if
and only if it has the following form:
\[
A=\left(\begin{array}{ccccccc}
0& a_{12}&a_{13}&a_{14}&\dots& a_{1,d-1}&a_{1d}\\
\dots& 0&a_{13}-a_{12}&a_{14}-a_{12}&\dots&
a_{1,d-1}-a_{12}&a_{1d}-a_{12}\\
\dots& \dots&0&a_{14}-a_{13}&\dots& a_{1,d-1}-a_{13}&a_{1d}-a_{13}\\
\dots& \dots&\dots&0&\dots& a_{1,d-1}-a_{14}&a_{1d}-a_{14}\\
\dots& \dots&\dots&\dots& \dots&\dots&\dots\\
\dots& \dots&\dots&\dots&\dots& 0&a_{1d}-a_{1,d-1}\\
\dots& \dots&\dots&\dots&\dots& \dots&0
\end{array}\right),
\]
where the subdiagonal elements are defined by skew-symmetry and so
are ignored.\qed
\end{remark}

\begin{lemma}\label{L-Fixbasis} The matrix
\begin{equation}\label{E-defP}
P=\sum_{2\le j\le d}\bse_{1j}\bse_{1j}^{T}+
\sum_{2\le i<j\le d}\bse_{ij}(\bse_{1j}^{T}-\bse_{1i}^{T})
\end{equation}
is a projector on the subspace $\mathbb{F}$. The vectors
\[
\bsf_{i}=P\bse_{1i}=\sum_{1\le j<i}\bse_{ji}-\sum_{i<j\le d}\bse_{ij},\quad
i=2,3,\ldots,d,
\]
where the second sum is assumed to vanish when $i=d$, form a basis
of $\mathbb{F}$.
\end{lemma}

\begin{proof}
Follows from Lemma~\ref{Lem-fixpoints} and Remark~\ref{Rem-fixmat}.
\end{proof}

Let us define in the space $\mathbb{R}^{d(d-1)/2}$ a new system of
coordinates by passing from the vectors $\bsx=(x_{ij})^{T}$ to the vectors
$\bsy=(y_{ij})^{T}$ in accordance with the rule
\begin{equation}\label{E-yx}
y_{ij}=\begin{cases}
x_{ij}&\textrm{for}\quad 1=i<j\le d;\\
x_{ij}+x_{1i}-x_{1j}&\textrm{for}\quad 2\le i<j\le d.
\end{cases}
\end{equation}
Then the inverse change of variables is defined as
\begin{equation}\label{E-xy}
x_{ij}=\begin{cases}
y_{ij}&\textrm{for}\quad 1=i<j\le d;\\
y_{ij}-y_{1i}+y_{1j}&\textrm{for}\quad 2\le i<j\le d.
\end{cases}
\end{equation}
Remark, that in terms of the variables $\bsy$, the subspace $\mathbb{F}$ of
common fixed points of the matrices $\{B_{\omega}\}$ can be characterised as
follows:
\[
\mathbb{F}=\left\{\bsy=(y_{ij}):~ y_{ij}=0~\textrm{for}~2\le i<j\le
d\right\}.
\]

Relations \eqref{E-yx} and \eqref{E-xy} can be rewritten in matrix form
\[
\bsy=Q^{-1}\bsx,\quad \bsx=Q\bsy,
\]
where
\begin{equation}\label{E-Q^{-1}}
Q=I+\sum_{2\le i<j\le d}\bse_{ij}(\bse_{1j}^{T}-\bse_{1i}^{T}),\quad
Q^{-1}=I-\sum_{2\le i<j\le d}\bse_{ij}(\bse_{1j}^{T}-\bse_{1i}^{T}).
\end{equation}
Since the identity matrix $I$ allows the representation
\[
I=\sum_{1\le i<j\le d}\bse_{ij}\bse_{ij}^{T}=
\sum_{2\le j\le d}\bse_{1j}\bse_{1j}^{T}+\sum_{2\le i<j\le
d}\bse_{ij}\bse_{ij}^{T},
\]
then by \eqref{E-defP} we have
\[
Q=\sum_{2\le i<j\le d}\bse_{ij}\bse_{ij}^{T}+P.
\]
Therefore
\[
\bsf_{i}=P\bse_{1i}=Q\bse_{1i},\quad i=2,3,\ldots, d,
\]
and then
\[
\bse_{1i}=Q^{-1}\bsf_{i},\quad i=2,3,\ldots, d.
\]
Because each of the vectors $\bsf_{i}$, $i=2,3,\ldots, d$, is a fixed point
of each of the matrices $B_{\omega}$, i.e.,
\[
\bsf_{i}=B_{\omega}\bsf_{i},\quad i=2,3,\ldots, d, ~\forall~\omega,
\]
then each of the matrices $Q^{-1}B_{\omega}Q$ takes the
block-triangular form:
\begin{equation}\label{E-QBQ}
    D_{\omega}:=Q^{-1}B_{\omega}Q=\left(\begin{array}{cc}
I& F_{\omega}\\
0& G_{\omega}
\end{array}\right),
\end{equation}
where
\[
\bse_{1i}=D_{\omega}\bse_{1i},\quad i=2,3,\ldots, d, ~\forall~\omega,
\]

To describe the structure of the matrices $D_{\omega}$ in more detail, make
use of the representations \eqref{E-defBx} and \eqref{E-Q^{-1}}, and of the
fact that $\bse_{ij}^{T}\bse_{mn}=1$ if and only if $i=m$ and $j=n$, whereas
$\bse_{ij}^{T}\bse_{mn}=0$ in other cases.

\paragraph{Case $1=k<i<j$}
By \eqref{E-defBx}, in this case
$B_{\omega}=I-\bse_{ij}(\bse_{ij}^{T}+\bse_{1i}^{T}-\bse_{1j}^{T})$, from
which
\begin{multline*}
D_{\omega}=Q^{-1}B_{\omega}Q\\
=\Bigl(I-\sum_{2\le m<n\le
d}\bse_{mn}(\bse_{1n}^{T}-\bse_{1m}^{T})\Bigr)
    \Bigl(I-\bse_{ij}(\bse_{ij}^{T}+\bse_{1i}^{T}-\bse_{1j}^{T})\Bigr)
    \Bigl(I+\sum_{2\le m<n\le
    d}\bse_{mn}(\bse_{1n}^{T}-\bse_{1m}^{T})\Bigr)\\
=I-\Bigl(I-\sum_{2\le m<n\le
d}\bse_{mn}(\bse_{1n}^{T}-\bse_{1m}^{T})\Bigr)
    \Bigl(\bse_{ij}(\bse_{ij}^{T}+\bse_{1i}^{T}-\bse_{1j}^{T})\Bigr)
    \Bigl(I+\sum_{2\le m<n\le
    d}\bse_{mn}(\bse_{1n}^{T}-\bse_{1m}^{T})\Bigr)\\
=I -\bse_{ij}(\bse_{ij}^{T}+\bse_{1i}^{T}-\bse_{1j}^{T})
    \Bigl(I+\sum_{2\le m<n\le
    d}\bse_{mn}(\bse_{1n}^{T}-\bse_{1m}^{T})\Bigr)\\
=I -\bse_{ij}(\bse_{ij}^{T}+\bse_{1i}^{T}-\bse_{1j}^{T})-
    \bse_{ij}(\bse_{ij}^{T}+\bse_{1i}^{T}-\bse_{1j}^{T})\sum_{2\le m<n\le
    d}\bse_{mn}(\bse_{1n}^{T}-\bse_{1m}^{T})\\
=I-\bse_{ij}(\bse_{ij}^{T}+\bse_{1i}^{T}-\bse_{1j}^{T})-\bse_{ij}(\bse_{1j}^{T}-\bse_{1i}^{T})=
   I-\bse_{ij}\bse_{ij}^{T}.
\end{multline*}

\paragraph{Case $1<k<i<j$}
By \eqref{E-defBx}, in this case
$B_{\omega}=I-\bse_{ij}(\bse_{ij}^{T}+\bse_{ki}^{T}-\bse_{kj}^{T})$, from
which
\begin{multline*}
D_{\omega}=Q^{-1}B_{\omega}Q\\
=\Bigl(I-\sum_{2\le m<n\le
d}\bse_{mn}(\bse_{1n}^{T}-\bse_{1m}^{T})\Bigr)
    \Bigl(I-\bse_{ij}(\bse_{ij}^{T}+\bse_{ki}^{T}-\bse_{kj}^{T})\Bigr)
    \Bigl(I+\sum_{2\le m<n\le
    d}\bse_{mn}(\bse_{1n}^{T}-\bse_{1m}^{T})\Bigr)\\
=I-\Bigl(I-\sum_{2\le m<n\le
d}\bse_{mn}(\bse_{1n}^{T}-\bse_{1m}^{T})\Bigr)
    \Bigl(\bse_{ij}(\bse_{ij}^{T}+\bse_{ki}^{T}-\bse_{kj}^{T})\Bigr)
    \Bigl(I+\sum_{2\le m<n\le
    d}\bse_{mn}(\bse_{1n}^{T}-\bse_{1m}^{T})\Bigr)\\
=I -\bse_{ij}(\bse_{ij}^{T}+\bse_{ki}^{T}-\bse_{kj}^{T})
    \Bigl(I+\sum_{2\le m<n\le
    d}\bse_{mn}(\bse_{1n}^{T}-\bse_{1m}^{T})\Bigr)\\
=I-\bse_{ij}(\bse_{ij}^{T}+\bse_{ki}^{T}-\bse_{kj}^{T})-\bse_{ij}(\bse_{1j}^{T}-\bse_{1i}^{T})
-\bse_{ij}(\bse_{1i}^{T}-\bse_{1k}^{T})+\bse_{ij}(\bse_{1j}^{T}-\bse_{1k}^{T})\\
=I-\bse_{ij}(\bse_{ij}^{T}+\bse_{ki}^{T}-\bse_{kj}^{T})=B_{\omega}.
\end{multline*}

\paragraph{Case $1= i<k<j$}
By \eqref{E-defBx}, in this case
$B_{\omega}=I-\bse_{1j}(\bse_{1j}^{T}-\bse_{1k}^{T}-\bse_{kj}^{T})$, from
which
\begin{multline*}
D_{\omega}=Q^{-1}B_{\omega}Q\\
=\Bigl(I-\sum_{2\le m<n\le
d}\bse_{mn}(\bse_{1n}^{T}-\bse_{1m}^{T})\Bigr)
    \Bigl(I-\bse_{1j}(\bse_{1j}^{T}-\bse_{1k}^{T}-\bse_{kj}^{T})\Bigr)
    \Bigl(I+\sum_{2\le m<n\le
    d}\bse_{mn}(\bse_{1n}^{T}-\bse_{1m}^{T})\Bigr)\\
=I-\Bigl(I-\sum_{2\le m<n\le
d}\bse_{mn}(\bse_{1n}^{T}-\bse_{1m}^{T})\Bigr)
    \Bigl(\bse_{1j}(\bse_{1j}^{T}-\bse_{1k}^{T}-\bse_{kj}^{T})\Bigr)
    \Bigl(I+\sum_{2\le m<n\le
    d}\bse_{mn}(\bse_{1n}^{T}-\bse_{1m}^{T})\Bigr)\\
=I-\Bigl(I-\sum_{2\le m<n\le
d}\bse_{mn}(\bse_{1n}^{T}-\bse_{1m}^{T})\Bigr)
    \Bigl(\bse_{1j}(\bse_{1j}^{T}-\bse_{1k}^{T}-\bse_{kj}^{T})-
    \bse_{1j}(\bse_{1j}^{T}-\bse_{1k}^{T})\Bigr)\\
=I+\Bigl(I-\sum_{2\le m<n\le
d}\bse_{mn}(\bse_{1n}^{T}-\bse_{1m}^{T})\Bigr)
    \bse_{1j}\bse_{kj}^{T}
=I+\Bigl(\bse_{1j}+\sum_{j<n}\bse_{jn}-\sum_{2\le
m<j}\bse_{mj}\Bigr) \bse_{kj}^{T}.
\end{multline*}

\paragraph{Case $1< i<k<j$}
By \eqref{E-defBx}, in this case
$B_{\omega}=I-\bse_{ij}(\bse_{ij}^{T}-\bse_{ik}^{T}-\bse_{kj}^{T})$, from
which
\begin{multline*}
D_{\omega}=Q^{-1}B_{\omega}Q\\
=\Bigl(I-\sum_{2\le m<n\le
d}\bse_{mn}(\bse_{1n}^{T}-\bse_{1m}^{T})\Bigr)
    \Bigl(I-\bse_{ij}(\bse_{ij}^{T}-\bse_{ik}^{T}-\bse_{kj}^{T})\Bigr)
    \Bigl(I+\sum_{2\le m<n\le
    d}\bse_{mn}(\bse_{1n}^{T}-\bse_{1m}^{T})\Bigr)\\
=I-\Bigl(I-\sum_{2\le m<n\le
d}\bse_{mn}(\bse_{1n}^{T}-\bse_{1m}^{T})\Bigr)
    \Bigl(\bse_{ij}(\bse_{ij}^{T}-\bse_{ik}^{T}-\bse_{kj}^{T})\Bigr)
    \Bigl(I+\sum_{2\le m<n\le
    d}\bse_{mn}(\bse_{1n}^{T}-\bse_{1m}^{T})\Bigr)\\
=I -\bse_{ij}(\bse_{ij}^{T}-\bse_{ik}^{T}-\bse_{kj}^{T})
    \Bigl(I+\sum_{2\le m<n\le
    d}\bse_{mn}(\bse_{1n}^{T}-\bse_{1m}^{T})\Bigr)\\
=I-\bse_{ij}(\bse_{ij}^{T}-\bse_{ik}^{T}-\bse_{kj}^{T})-\bse_{ij}(\bse_{1j}^{T}-\bse_{1i}^{T})
+\bse_{ij}(\bse_{1k}^{T}-\bse_{1i}^{T})+\bse_{ij}(\bse_{1j}^{T}-\bse_{1k}^{T})\\
=I-\bse_{ij}(\bse_{ij}^{T}-\bse_{ik}^{T}-\bse_{kj}^{T})=B_{\omega}.
\end{multline*}

\paragraph{Case $1=i<j<k$}
By \eqref{E-defBx}, in this case
$B_{\omega}=I-\bse_{1j}(\bse_{1j}^{T}-\bse_{1k}^{T}+\bse_{jk}^{T})$, from
which
\begin{multline*}
D_{\omega}=Q^{-1}B_{\omega}Q\\
=\Bigl(I-\sum_{2\le m<n\le
d}\bse_{mn}(\bse_{1n}^{T}-\bse_{1m}^{T})\Bigr)
    \Bigl(I-\bse_{1j}(\bse_{1j}^{T}-\bse_{1k}^{T}+\bse_{jk}^{T})\Bigr)
    \Bigl(I+\sum_{2\le m<n\le
    d}\bse_{mn}(\bse_{1n}^{T}-\bse_{1m}^{T})\Bigr)\\
=I-\Bigl(I-\sum_{2\le m<n\le
d}\bse_{mn}(\bse_{1n}^{T}-\bse_{1m}^{T})\Bigr)
    \Bigl(\bse_{1j}(\bse_{1j}^{T}-\bse_{1k}^{T}+\bse_{jk}^{T})\Bigr)
    \Bigl(I+\sum_{2\le m<n\le
    d}\bse_{mn}(\bse_{1n}^{T}-\bse_{1m}^{T})\Bigr)\\
=I-\Bigl(I-\sum_{2\le m<n\le
d}\bse_{mn}(\bse_{1n}^{T}-\bse_{1m}^{T})\Bigr)
    \Bigl(\bse_{1j}(\bse_{1j}^{T}-\bse_{1k}^{T}+\bse_{jk}^{T})+
    \bse_{1j}(\bse_{1k}^{T}-\bse_{1j}^{T})\Bigr)\\
=I-\Bigl(I-\sum_{2\le m<n\le
d}\bse_{mn}(\bse_{1n}^{T}-\bse_{1m}^{T})\Bigr)
    \bse_{1j}\bse_{jk}^{T}
=I-\Bigl(\bse_{1j}+\sum_{j<n}\bse_{jn}-\sum_{2\le
m<j}\bse_{mj}\Bigr) \bse_{jk}^{T}.
\end{multline*}

\paragraph{Case $1< i<j<k$}
By \eqref{E-defBx}, in this case
$B_{\omega}=I-\bse_{ij}(\bse_{ij}^{T}-\bse_{ik}^{T}+\bse_{jk}^{T})$, from
which
\begin{multline*}
D_{\omega}=Q^{-1}B_{\omega}Q\\
=\Bigl(I-\sum_{2\le m<n\le
d}\bse_{mn}(\bse_{1n}^{T}-\bse_{1m}^{T})\Bigr)
    \Bigl(I-\bse_{ij}(\bse_{ij}^{T}-\bse_{ik}^{T}+\bse_{jk}^{T})\Bigr)
    \Bigl(I+\sum_{2\le m<n\le
    d}\bse_{mn}(\bse_{1n}^{T}-\bse_{1m}^{T})\Bigr)\\
=I-\Bigl(I-\sum_{2\le m<n\le
d}\bse_{mn}(\bse_{1n}^{T}-\bse_{1m}^{T})\Bigr)
    \Bigl(\bse_{ij}(\bse_{ij}^{T}-\bse_{ik}^{T}+\bse_{jk}^{T})\Bigr)
    \Bigl(I+\sum_{2\le m<n\le
    d}\bse_{mn}(\bse_{1n}^{T}-\bse_{1m}^{T})\Bigr)\\
=I -\bse_{ij}(\bse_{ij}^{T}-\bse_{ik}^{T}+\bse_{jk}^{T})
    \Bigl(I+\sum_{2\le m<n\le
    d}\bse_{mn}(\bse_{1n}^{T}-\bse_{1m}^{T})\Bigr)\\
=I-\bse_{ij}(\bse_{ij}^{T}-\bse_{ik}^{T}+\bse_{jk}^{T})-\bse_{ij}(\bse_{1j}^{T}-\bse_{1i}^{T})
+\bse_{ij}(\bse_{1k}^{T}-\bse_{1i}^{T})-\bse_{ij}(\bse_{1k}^{T}-\bse_{1j}^{T})\\
=I-\bse_{ij}(\bse_{ij}^{T}-\bse_{ik}^{T}+\bse_{jk}^{T})=B_{\omega}.
\end{multline*}

It is convenient to unite the obtained relations into a single formula:
\begin{equation}\label{E-Domega}
D_{\omega}=\begin{cases}
I-\bse_{ij}\bse_{ij}^{T}&\textrm{for}\quad 1=k<i<j;\\
I-\bse_{ij}(\bse_{ij}^{T}+\bse_{ki}^{T}-\bse_{kj}^{T})=B_{\omega}&\textrm{for}\quad
1<k<i<j;\\
I+\Bigl(\bse_{1j}+\sum_{j<n}\bse_{jn}-\sum_{2\le
m<j}\bse_{mj}\Bigr) \bse_{kj}^{T}&\textrm{for}\quad 1= i<k<j;\\
I-\bse_{ij}(\bse_{ij}^{T}-\bse_{ik}^{T}-\bse_{kj}^{T})=B_{\omega}&\textrm{for}\quad
1< i<k<j;\\
I-\Bigl(\bse_{1j}+\sum_{j<n}\bse_{jn}-\sum_{2\le
m<j}\bse_{mj}\Bigr) \bse_{jk}^{T}&\textrm{for}\quad 1=i<j<k;\\
I-\bse_{ij}(\bse_{ij}^{T}-\bse_{ik}^{T}+\bse_{jk}^{T})=B_{\omega}&\textrm{for}\quad
1< i<j<k.\\
\end{cases}
\end{equation}

Now, to compute the matrix $F_{\omega}$, it is necessary to retain in the
right-hand part of \eqref{E-Domega} only the summands of the form
$\bse_{mn}\bse_{pq}^{T}$ in which $m=1$ and $n,p,q>1$. Similarly, to compute
the matrix $G_{\omega}$ it is necessary to retain in the right-hand part of
\eqref{E-Domega} only the summands of the form $\bse_{mn}\bse_{pq}^{T}$ in
which $m,n,p,q>1$. In doing so one should keep in mind that $I=\sum_{1\le
m<n\le d}\bse_{mn}\bse_{mn}^{T}$. This results in
\[
F_{\omega}=\begin{cases}
\hfill 0&\textrm{for}\quad 1\le k<i<j,\\
\hfill\bse_{1j}\bse_{kj}^{T}&\textrm{for}\quad 1= i<k<j,\\
\hfill0&\textrm{for}\quad 1< i<k<j,\\
\hfill-\bse_{1j}\bse_{jk}^{T}&\textrm{for}\quad 1=i<j<k,\\
\hfill0&\textrm{for}\quad 1< i<j<k
\end{cases}
\]
and
\begin{equation}\label{E-Gomega}
G_{\omega}=\begin{cases}
I-\bse_{ij}\bse_{ij}^{T}&\textrm{for}\quad 1= k<i<j,\\
I-\bse_{ij}(\bse_{ij}^{T}+\bse_{ki}^{T}-\bse_{kj}^{T})&\textrm{for}\quad 1<
k<i<j,\\
I+\Bigl(\sum_{j<n}\bse_{jn}-\sum_{2\le
m<j}\bse_{mj}\Bigr) \bse_{kj}^{T}&\textrm{for}\quad 1= i<k<j,\\
I-\bse_{ij}(\bse_{ij}^{T}-\bse_{ik}^{T}-\bse_{kj}^{T})&\textrm{for}\quad 1<
i<k<j,\\
I-\Bigl(\sum_{j<n}\bse_{jn}-\sum_{2\le
m<j}\bse_{mj}\Bigr) \bse_{jk}^{T}&\textrm{for}\quad 1=i<j<k,\\
I-\bse_{ij}(\bse_{ij}^{T}-\bse_{ik}^{T}+\bse_{jk}^{T})&\textrm{for}\quad 1<
i<j<k.
\end{cases}
\end{equation}
Naturally, in \eqref{E-Gomega}, in contrast to \eqref{E-Domega}, the
identity matrix $I$ acts in the subspace spanned over the vectors
$\bse_{mn}$, $2\le m<n\le d$, and therefore it is representable as
$I=\sum_{2\le m<n\le d}\bse_{mn}\bse_{mn}^{T}$.

\begin{remark}\label{Rem6}
In formulae \eqref{E-Domega} and \eqref{E-Gomega} the sums corresponding to
an empty set of indices are assumed to vanish. For instance,
$\sum_{j<n}\bse_{jn}=0$ if $j=d$ or $\sum_{2\le m<j}\bse_{mj}=0$ if $j=2$.
\qed
\end{remark}

\begin{remark}\label{Rem7}
The idea of highlighting the matrices $G_{\omega}$ is to lower the dimension
of the studied matrices since, according to \eqref{E-QBQ}, the behaviour of
the products of matrices $G_{\omega}$ essentially determines the properties
of the products of matrices $D_{\omega}$ and $B_{\omega}$. \qed
\end{remark}

\section{Case $d=4$}\label{S-d4}

The detailed analysis of the FX currency market with $4$ trader-arbitragers,
i.e., when $d=4$, was fulfilled in \cite{KozCalPok:ArXiv10} and
\cite{CrossKPP:MECA12}. Recall some related results. In this case the matrix
$A$ is of the form
\[
A=\left(\begin{array}{rrrr}
0_{\hphantom{22}}&a_{12}&a_{13}&a_{14}\\
-a_{12}&0_{\hphantom{22}}&a_{23}&a_{24}\\
-a_{13}&-a_{23}&0_{\hphantom{22}}&a_{34}\\
-a_{14}&-a_{14}&-a_{34}&0_{\hphantom{22}}
\end{array}\right).
\]

In terms of the upper off-diagonal elements, the arbitrage rules can be
written as follows:
\[
\left\{\begin{array}{cccrl}
(123):&a_{12,\textrm{new}}&=&a_{13}-a_{23},&\quad\textrm{other elements
unchanged},\\
(124):&a_{12,\textrm{new}}&=&a_{14}-a_{24},&\quad\textrm{other elements
unchanged},\\
(132):&a_{13,\textrm{new}}&=&a_{12}+a_{23},&\quad\textrm{other elements
unchanged},\\
(134):&a_{13,\textrm{new}}&=&a_{14}-a_{34},&\quad\textrm{other elements
unchanged},\\
(142):&a_{14,\textrm{new}}&=&a_{12}+a_{24},&\quad\textrm{other elements
unchanged},\\
(143):&a_{14,\textrm{new}}&=&a_{13}+a_{34},&\quad\textrm{other elements
unchanged},\\
(231):&a_{23,\textrm{new}}&=&-a_{12}+a_{13},&\quad\textrm{other elements
unchanged},\\
(234):&a_{23,\textrm{new}}&=&a_{24}-a_{34},&\quad\textrm{other elements
unchanged},\\
(241):&a_{24,\textrm{new}}&=&-a_{12}+a_{14},&\quad\textrm{other elements
unchanged},\\
(243):&a_{24,\textrm{new}}&=&a_{23}+a_{34},&\quad\textrm{other elements
unchanged},\\
(341):&a_{34,\textrm{new}}&=&-a_{13}+a_{14},&\quad\textrm{other elements
unchanged},\\
(342):&a_{34,\textrm{new}}&=&-a_{23}+a_{24},&\quad\textrm{other elements
unchanged}.
\end{array}\right.
\]

The subspace of fixed points $\mathbb{F}$ consists of all the vectors $\bsx
=(x_{ij})^{T}$, coordinates of which satisfy
\[
\left\{\begin{alignedat}{4}
x_{12}&=\hphantom{-}x_{13}-x_{23},\quad&
x_{12}&=x_{14}-x_{24},\quad&
x_{13}&=\hphantom{-}x_{12}+x_{23},\quad&
x_{13}&=\hphantom{-}x_{14}-x_{34},\\
x_{14}&=\hphantom{-}x_{12}+x_{24},\quad&
x_{14}&=x_{13}+x_{34},\quad&
x_{23}&=-x_{12}+x_{13},\quad&
x_{23}&=\hphantom{-}x_{24}-x_{34},\\
x_{24}&=-x_{12}+x_{14},\quad&
x_{24}&=x_{23}+x_{34},\quad&
x_{34}&=-x_{13}+x_{14},\quad&
x_{34}&=-x_{23}+x_{24}.
\end{alignedat}\right.
\]
This system has redundant equations. In particular, all the equations
containing the minus sign are redundant. By removing them we obtain:
\[
x_{13}=x_{12}+x_{23},\quad
x_{14}=x_{12}+x_{24},\quad
x_{14}=x_{13}+x_{34},\quad
x_{24}=x_{23}+x_{34}.
\]
Here one of the equations determining $x_{14}$ is also redundant. By
removing it, we obtain the final set of independent equations defining the
subspace $\mathbb{F}$:
\[
x_{13}=x_{12}+x_{23},\quad
x_{14}=x_{13}+x_{34},\quad
x_{24}=x_{23}+x_{34}.
\]
By choosing $x_{12}$, $x_{13}$ and $x_{14}$ as independent variables, the
remaining variables can be defined by the equalities
\[
x_{23}=-x_{12}+x_{13},\quad
x_{24}=-x_{12}+x_{14},\quad x_{34}=-x_{13}+x_{14}.
\]

The form of the matrices $B_{\omega}, D_{\omega}$ and $G_{\omega}$ for this
case is given in \ref{A-d4}.

As was noted in \cite{KozCalPok:ArXiv10} and \cite{CrossKPP:MECA12}, in the
case $d=4$ \textsl{the arbitrage process may be non-convergent for some
periodical sequence of the arbitrage rules}. To support this claim let us
consider the set of all vectors each of which is mapped by some of the
matrices $G_{\omega}$ to zero. Direct calculations show that, up to
non-negative factors, this set consists of the following $12$ vectors:
\[
\bss_{1}=\left(\begin{array}{r}1\\1\\0\end{array}\right),~
\bss_{2}=\left(\begin{array}{r}1\\0\\-1\end{array}\right),~
\bss_{3}=\left(\begin{array}{r}0\\1\\1\end{array}\right),~
\bss_{4}=\left(\begin{array}{r}1\\0\\0\end{array}\right),~
\bss_{5}=\left(\begin{array}{r}0\\1\\0\end{array}\right),~
\bss_{6}=\left(\begin{array}{r}0\\0\\1\end{array}\right),
\]
\[
\bss_{7}=-\bss_{1},\quad \bss_{8}=-\bss_{2},\quad
\bss_{9}=-\bss_{3},\quad \bss_{10}=-\bss_{4},\quad
\bss_{11}=-\bss_{5},\quad \bss_{12}=-\bss_{6},
\]
Moreover, as is easily verified, the set
\[
\mathbb{S}=\textrm{co}\left\{\pm \bss_{1},\pm \bss_{2},\pm \bss_{3},\pm
\bss_{4},
\pm \bss_{5},\pm \bss_{6}\right\}
\]
is mapped by each of the matrices $G_{\omega}$ into itself and is a body,
i.e., contains interior points. Hence, $\mathbb{S}$ is a unit ball of some
norm $\|\cdot\|_{*}$ for which
\begin{equation}\label{E-Gnorm1}
\|G_{\omega}\|_{*}\le 1,\quad\forall\omega.
\end{equation}

So, the set of matrices $\{G_{\omega}\}$ is \textsl{neutrally stable} and,
it seems, there is the possibility that all products of matrices
$G_{\omega}$ are convergent. In Sections~\ref{Ex1} and \ref{Ex2} it is shown
that this is not the case.


\begin{remark}\label{R-linrate}
By \eqref{E-QBQ} the following representation is valid
\[
B_{\omega}=Q\left(\begin{array}{cc}
I& F_{\omega}\\
0& G_{\omega}
\end{array}\right)Q^{-1}.
\]
From here, and from the estimate~\eqref{E-Gnorm1}, it then follows that the
norm of any products of the matrices $B_{\omega}$ may grow no faster than
$cn$ with some constant $c$, where $n$ is the number of factors $B_{\omega}$
in the product. Example~\ref{Ex2} demonstrates that such a rate of growth of
products of the matrices $B_{\omega}$, in the case $d=4$, is really
achievable. \qed
\end{remark}

\subsection{Example}\label{Ex1}

Consider the following sequence of the arbitrage rules
\[
\begin{array}{rrrr}
\omega(1)=(1,4,2),&
\omega(2)=(1,2,3),&
\omega(3)=(3,4,1),&
\omega(4)=(1,4,2),\\
\omega(5)=(1,3,4),&
\omega(6)=(2,4,3),&
\omega(7)=(2,3,1),&
\omega(8)=(3,4,2),\\
\omega(9)=(2,4,1),&
\omega(10)=(1,3,4),&
\omega(11)=(3,4,2),&
\omega(12)=(1,4,3),\\
\omega(13)=(2,3,4),&
\omega(14)=(1,3,2),&
\omega(15)=(1,2,4),&
\omega(16)=(1,4,3),
\end{array}
\]
and extend it by periodicity. Set also $\bsx_{0}=\bss_{1}$, and build
recurrently the sequence
\[
\bsx_{n}=G_{\omega(n)}\bsx_{n-1},\quad n=1,2,\ldots,16.
\]
Then
\begin{alignat*}{3}
\bsx_{0}&=\bss_{1},&\quad
\bsx_{1}&=G_{\omega(1)}\bsx_{0}=G_{(142)}\bsx_{0}=\bss_{2},\\
\bsx_{2}&=G_{\omega(2)}\bsx_{1}=G_{(123)}\bsx_{1}=-\bss_{3},&\quad
\bsx_{3}&=G_{\omega(3)}\bsx_{2}=G_{(341)}\bsx_{2}=\bss_{5},\\
\bsx_{4}&=G_{\omega(4)}\bsx_{3}=G_{(142)}\bsx_{3}=\bss_{6},&\quad
\bsx_{5}&=G_{\omega(5)}\bsx_{4}=G_{(134)}\bsx_{4}=\bss_{4},\\
\bsx_{6}&=G_{\omega(6)}\bsx_{5}=G_{(243)}\bsx_{5}=\bss_{1},&\quad
\bsx_{7}&=G_{\omega(7)}\bsx_{6}=G_{(231)}\bsx_{6}=-\bss_{5},\\
\bsx_{8}&=G_{\omega(8)}\bsx_{7}=G_{(342)}\bsx_{7}=\bss_{3},&\quad
\bsx_{9}&=G_{\omega(9)}\bsx_{8}=G_{(241)}\bsx_{8}=\bss_{6},\\
\bsx_{10}&=G_{\omega(10)}\bsx_{9}=G_{(134)}\bsx_{9}=\bss_{4},&\quad
\bsx_{11}&=G_{\omega(11)}\bsx_{10}=G_{(342)}\bsx_{10}=\bss_{2},\\
\bsx_{12}&=G_{\omega(12)}\bsx_{11}=G_{(143)}\bsx_{11}=\bss_{1},&\quad
\bsx_{13}&=G_{\omega(13)}\bsx_{12}=G_{(234)}\bsx_{12}=\bss_{1},\\
\bsx_{14}&=G_{\omega(14)}\bsx_{13}=G_{(132)}\bsx_{13}=\bss_{3},&\quad
\bsx_{15}&=G_{\omega(15)}\bsx_{14}=G_{(124)}\bsx_{14}=-\bss_{2},\\
\bsx_{16}&=G_{\omega(16)}\bsx_{15}=G_{(143)}\bsx_{15}=-\bss_{1}.
\end{alignat*}

Continuing building, we will obtain a $32$-periodic sequence $\bsx_{n}$. It
is not difficult to see that the sequence of products of the corresponding
matrices $B_{\omega}$:
\[
H_{n}=B_{\omega(n)}H_{n-1},\qquad H_{0}=I,~ n=1,2,\ldots\,,
\]
will be also \textbf{$32$-periodic} and \textbf{non-convergent to any
limit}.

\subsection{Example}\label{Ex2}

Consider the following sequence of the arbitrage rules
\[
\begin{array}{rrrr}
\omega(1)=(1,4,3),&
\omega(2)=(3,4,1),&
\omega(3)=(3,4,2),&
\omega(4)=(1,4,2),\\
\omega(5)=(1,2,4),&
\omega(6)=(2,3,1),&
\omega(7)=(1,3,2),&
\omega(8)=(2,4,3),\\
\omega(9)=(1,3,4),&
\omega(10)=(2,4,1),&
\omega(11)=(1,2,3),&
\omega(12)=(2,3,4),
\end{array}
\]
and extend it by periodicity. Set also $\bsx_{0}=\bss_{1}$, and build
recurrently the sequence
\[
\bsx_{n}=G_{\omega(n)}\bsx_{n-1},\quad n=1,2,\ldots,12.
\]
Then
\begin{alignat*}{2}
\bsx_{0}&=\bss_{1},&\quad
\bsx_{1}&=G_{\omega(1)}\bsx_{0}=G_{(143)}\bsx_{0}=\bss_{1},\\
\bsx_{2}&=G_{\omega(2)}\bsx_{1}=G_{(341)}\bsx_{1}=\bss_{1},&\quad
\bsx_{3}&=G_{\omega(3)}\bsx_{2}=G_{(342)}\bsx_{2}=\bss_{1},\\
\bsx_{4}&=G_{\omega(4)}\bsx_{3}=G_{(142)}\bsx_{3}=\bss_{2},&\quad
\bsx_{5}&=G_{\omega(5)}\bsx_{4}=G_{(124)}\bsx_{4}=\bss_{2},\\
\bsx_{6}&=G_{\omega(6)}\bsx_{5}=G_{(231)}\bsx_{5}=-\bss_{6},&\quad
\bsx_{7}&=G_{\omega(7)}\bsx_{6}=G_{(132)}\bsx_{6}=-\bss_{6},\\
\bsx_{8}&=G_{\omega(8)}\bsx_{7}=G_{(243)}\bsx_{7}=-\bss_{3},&\quad
\bsx_{9}&=G_{\omega(9)}\bsx_{8}=G_{(134)}\bsx_{8}=-\bss_{1},\\
\bsx_{10}&=G_{\omega(10)}\bsx_{9}=G_{(241)}\bsx_{9}=-\bss_{4},&\quad
\bsx_{11}&=G_{\omega(11)}\bsx_{10}=G_{(123)}\bsx_{10}=-\bss_{5},\\
\bsx_{12}&=G_{\omega(12)}\bsx_{11}=G_{(234)}\bsx_{11}=\bss_{1}.
\end{alignat*}

Continuing building, we will obtain a $12$-periodic sequence $\bsx_{n}$. The
sequence of the corresponding products of the matrices $G_{\omega(n)}$ is
\textbf{bounded} by \eqref{E-Gnorm1} whereas the sequence of products of the
matrices $B_{\omega}$:
\[
H_{n}=B_{\omega(n)}H_{n-1},\qquad H_{0}=I,~ n=1,2,\ldots\,,
\]
is \textbf{divergent}, and the norms of matrices $H_{n}$ \textbf{have the
linear rate of growth}, see Remark~\ref{R-linrate}.

\section{Case $d=5$}\label{S-d5}

The full set of the matrices $B_{\omega}, D_{\omega}$ and
$G_{\omega}$ for this case is presented in \ref{A-d5}.
Direct calculations show that the matrix
\[
H=B_{(143)}B_{(231)}B_{(245)}B_{(342)}B_{(451)}B_{(124)}B_{(453)}
\]
is as follows
{\renewcommand{\arraystretch}{0.7}%
\[
H = \left(\begin{array}{@{}R@{}R@{}R@{}R@{}R@{}R@{}R@{}R@{}R@{}R@{}}
0 & 0 & 1 & 0 & 0 & -1 & 0 & 0 & 0 & 0 \\
0 & 1 & 0 & 0 & 0 & 0 & 0 & 0 & 0 & 0 \\
0 & 1 & 0 & 0 & -1 & 1 & 0 & 0 & 0 & 0 \\
0 & 0 & 0 & 1 & 0 & 0 & 0 & 0 & 0 & 0 \\
0 & 1 & -1 & 0 & 0 & 1 & 0 & 0 & 0 & 0 \\
0 & 0 & 1 & -1 & 0 & 0 & 1 & 0 & 0 & 0 \\
0 & 0 & 0 & 0 & 0 & 0 & 1 & 0 & 0 & 0 \\
0 & 0 & 0 & 0 & -1 & 1 & 0 & 0 & 0 & 0 \\
0 & 0 & 0 & 0 & 0 & 0 & 0 & 0 & 1 & 0 \\
0 & 0 & -1 & 1 & 0 & 0 & 0 & 0 & 0 & 0
\end{array}\right),
\]}
and its spectral radius $\rho(H)$ is strictly greater than $1$:
\[
\rho(H)=\frac{1+\sqrt{5}}{2}\simeq 1.618.
\]
To ascertain this it suffices to note that the matrix $H$ has the following
eigenvalues (ordered by decreasing absolute values):
\[
-\frac{1+\sqrt{5}}{2},~1,~1,~1,~1,~1,~\frac{\sqrt{5}-1}{2},~0,~0,~0.
\]

Consider the following sequence of the arbitrage rules
\[
\begin{array}{rrrr}
\omega(1)=(4,5,3),&
\omega(2)=(1,2,4),&
\omega(3)=(4,5,1),&
\omega(4)=(3,4,2),\\
\omega(5)=(2,4,5),&
\omega(6)=(2,3,1),&
\omega(7)=(1,4,3),
\end{array}
\]
and extend it by periodicity. Then, as is known from the linear algebra, for
almost all initial vectors\footnote{To be more specific, for all initial
vectors not belonging to the linear invariant subspace of the matrix $H$
corresponding to the set of eigenvalues distinct from
$-\frac{1+\sqrt{5}}{2}$.} $\bsx_{0}$, the sequence of vectors
\[
\bsx_{n}=B_{\omega(n)}\bsx_{n-1},\quad n=1,2,\ldots\,,
\]
obtained by sequential application of the arbitrage rules $\omega(n)$ will
diverge with the rate $\left(\frac{1+\sqrt{5}}{2}\right)^{n/7}$, i.e.,
\[
c\, \left(\frac{1+\sqrt{5}}{2}\right)^{n/7}\|\bsx_{0}\|\le
\|\bsx_{n}\| \le C\, \left(\frac{1+\sqrt{5}}{2}\right)^{n/7}\|\bsx_{0}\|,
\]
where $c, C$ are some positive constants. Taking into account that $1.071\le
\left(\frac{1+\sqrt{5}}{2}\right)^{1/7}\le 1.072$, the inequalities obtained
may be represented in a more descriptive form:
\begin{equation}\label{E-exprate}
c\, 1.071^{n}\|\bsx_{0}\|\le\|\bsx_{n}\| \le C\, 1.072^{n}\|\bsx_{0}\|.
\end{equation}

\begin{remark}\label{Rem-exp}
Relations \eqref{E-exprate} imply that for $d=5$, in contrast to the case
$d=4$, the norms of the vectors $\{\bsx_{n}\}$ as well as the norms of the
matrix products
\[
H_{n}=B_{\omega(n)}H_{n-1},\qquad H_{0}=I,~ n=1,2,\ldots\,,
\]
\textbf{may grow with the exponential rate}. \qed
\end{remark}

\begin{remark}\label{Rem-inherit}
As was shown in Remark~\ref{Rem-exp}, in the case $d=5$, the norms of the
matrix products may grow with the exponential rate whereas in the case
$d=4$, according to Section~\ref{S-d4}, the norms of the matrix products may
have at most a linear rate of growth. Any sequence of arbitrage rules
involving $4$ traders may be treated as a particular case of a sequence of
arbitrage rules for $5$ traders. Therefore in the case $d=5$ one may observe
also the sequences of matrices whose norms of the products have a linear
rate of growth, as well as sequences of matrices (vectors) varying
periodically, see Section~\ref{S-d4}.

When $d>5$ we face an effect of ``inheritance of instability'' taking place
in the cases $d=4,5$. In this case one may observe all the types of
behaviour of the matrix products described above as well as of the
corresponding state vectors $\bsx_{n}$: periodicity, growth at a linear rate
and growth at an exponential rate. \qed
\end{remark}

\begin{remark}\label{Rem-superexp}
If we return back from the ``artificial'' quantities $a_{ij}$ to the true
exchange rates
\begin{equation}\label{E-rviaa}
r_{ij}=e^{a_{ij}},
\end{equation}
we arrive to a rather unexpected conclusion.

In the case $d=4$ the exchange rates may vary periodically or grow
exponentially which follows from \eqref{E-rviaa} and the results of
Section~\ref{S-d4}.

In the case $d\ge 5$ the exchange rates also may vary periodically or grow
exponentially but also, what is surprising, they \textbf{may grow in
accordance with the double exponential law}. For instance, in the case $d=5$
formula \eqref{E-exprate} implies
\[
e^{\tilde{c}\, 1.071^{n}}\le\|\bsr_{n}\|\le e^{\tilde{C}\, 1.072^{n}},
\]
where $\bsr_{n}$ is the exchange rate vector
\[
\left\{r_{12}, r_{13},\ldots, r_{1d}, r_{23},
r_{24},\ldots\,r_{2d},\ldots,r_{d-1,d}\right\}^{T}
\]
at the moment $n$. \qed
\end{remark}

\begin{remark}\label{Rem-fin}
As follows from the analysis undertaken, after introducing new variables
bringing the matrices $B_{\omega}$ to block-triangular form, one obtains a
set of variables, which is updated only by matrices $G_{\omega}$, and the
behaviour of which is thus determined only by \textbf{its own history}.
These variables may be called, in a sense, ``leading exchange rates
indices.'' For $d=4$, this set is constituted of the variables
\begin{equation}\label{E-newvar4}
a_{12}-a_{13}+a_{23},\quad
a_{12}-a_{14}+a_{24},\quad
a_{13}-a_{14}+a_{34},
\end{equation}
which, in terms of exchange rates, corresponds to the set of variables
\begin{equation}\label{E-newvar4x}
\frac{r_{12}\cdot r_{23}}{r_{13}}=r_{12}\cdot r_{23}\cdot r_{31},\quad
\frac{r_{12}\cdot r_{24}}{r_{14}}=r_{12}\cdot r_{24}\cdot r_{41},\quad
\frac{r_{13}\cdot r_{34}}{r_{14}}=r_{13}\cdot r_{34}\cdot r_{41}.
\end{equation}
For $d=5$, the related set is constituted of the variables
\begin{equation}\label{E-newvar5}
a_{12}-a_{13}+a_{23},\quad
a_{12}-a_{14}+a_{24},\quad
a_{12}-a_{15}+a_{25},\quad
a_{13}-a_{14}+a_{34},\quad
a_{13}-a_{15}+a_{35},\quad
a_{14}-a_{15}+a_{45},
\end{equation}
and to it, in terms of exchange rates, corresponds the set of variables
\begin{equation}\label{E-newvar5x}
r_{12}\cdot r_{23}\cdot r_{31},\quad
r_{12}\cdot r_{24}\cdot r_{41},\quad
r_{12}\cdot r_{25}\cdot r_{51},\quad
r_{13}\cdot r_{34}\cdot r_{41},\quad
r_{13}\cdot r_{35}\cdot r_{51},\quad
r_{14}\cdot r_{45}\cdot r_{51}.
\end{equation}

For $d=4$, the variables \eqref{E-newvar4} and \eqref{E-newvar4x} always
\textbf{vary boundedly} over time, while, for $d\ge5$, the variables
\eqref{E-newvar5} and \eqref{E-newvar5x} may grow unboundedly. At the same
time, the original set of exchange rates $\{r_{ij}\}$ and the corresponding
set of ``artificial'' log-variables $\{a_{ij}\}$ for $d=4$, as well as for
$d\ge5$, \textbf{may grow unboundedly}.
\end{remark}

The intuition underlying the ``leading exchange rates indices'' in
\eqref{E-newvar5x} can be explained as follows. The term $r_{12} r_{23}
r_{31}$ can be interpreted in the following way: the $1$st trader is selling
the currency 1 for the currency 2, then uses the currency 2 to buy currency
3 and then uses the currency 3 to buy his ``home currency'' 1. So, the term
$r_{12} r_{23} r_{31}$ is the ``arbitrage profit factor'' for the $1$st
trader when the trades involve the ``most simple'' cyclic triangular path:
$$
1\rightarrow 2\rightarrow 3\rightarrow 1.
$$

If this factor $r_{12} r_{23} r_{31}$ is equals $1$ then such ``cyclic
triangular trades'' bring no profit to the $1$st trader, while when $r_{12}
r_{23} r_{31}\neq 1$ then such ``cyclic triangular trades'' bring profit (or
loss) to the $1$st trader. Hence the condition
\[
r_{12} r_{23} r_{31}=1
\]
is the ``no-arbitrage'' condition for trading by the cyclic path
$1\rightarrow 2\rightarrow 3\rightarrow 1$, \textbf{for the $1$st trader}.
This ``unit circle'' no-arbitrage condition was pointed out in \citet[first
published in 1838]{cournot29}.

Then \eqref{E-newvar5x} is simply the full set of factors whose simultaneous
equality to $1$ represents the no-arbitrage conditions (NACs) for the whole
set of currencies in the $d=5$ case:
\begin{equation}\label{No-Arb-Cond}
r_{12} r_{23} r_{31}=1,\quad
r_{12} r_{24} r_{41}=1,\quad
r_{12} r_{25} r_{51}=1,\quad
r_{13} r_{34} r_{41}=1,\quad
r_{13} r_{35} r_{51}=1,\quad
r_{14} r_{45} r_{51}=1.
\end{equation}

It appears that different currencies are represented in \eqref{No-Arb-Cond}
non-symmetrically. But in fact there are different NACs which are equivalent
to each other.

We may introduce a measure, say,
  \[
  Q=|r_{12} r_{23} r_{31}-1|+|r_{12} r_{24} r_{41}-1|+
|r_{12} r_{25} r_{51}-1|+
|r_{13} r_{34} r_{41}-1|+
|r_{13} r_{35} r_{51}-1|+
|r_{14} r_{45} r_{51}-1|,
  \]
which can be interpreted as an ``accumulative'' NAC. The meaning is that the
equality $Q=0$ involves the absence of profitable arbitrage opportunities in
the whole FX market system, while the quantity $Q\neq0$ characterises the
extent of profitable arbitrage opportunities in the FX market system.

\section{Concluding Remarks}\label{S-ConcRem}

The EMH implies that arbitrage operations in markets such as those for FX
should be self-extinguishing. This simplistic view of arbitrage processes
may reflect a failure to consider the complex chains of arbitrage possible
in financial markets, such as those for FX, that have segmented trading
platforms and fragmented information. In \cite{KozCalPok:ArXiv10} and
\cite{CrossKPP:MECA12} it was shown, applying combinatorial analysis and
asynchronous systems theory, that arbitrage sequences in a $4$-currency
world tend to display periodicity or exponential growth rather than being
self-extinguishing. The present paper has used a log-linear reformulation of
the analysis to show that in a $(d\ge5)$-currency world, as well as
inheriting the periodicity or exponential growth found in the $4$-currency
case, arbitrage sequences can be augmented by a double-exponential process.
Thus there appears to be a cumulative ``inheritance of instability''.

The results may help to explain why extensive arbitrage trading is endemic
in FX markets. The fact that the number of currencies makes a difference is
interesting: periodicity or exponential behaviour arise when you move from
$3$ to $4$ currencies; double exponential behaviour as well, when you move
beyond $4$ currencies. The \textsl{Financial Times} provides daily quotes
for $52$ currencies against the US dollar, euro and pound sterling. But some
currencies are more important than others. BIS data indicate that trades
involving the US dollar accounted for $85/200$ of all FX trades in 2010,
followed by the euro $(39/200)$, Japanese yen $(19/200)$ and sterling
$(13/200)$, so trades involving these four key currencies accounted for
$78\%$ of global FX turnover \citep[p.~12]{BIS}. In terms of currency pairs
the leading trades are the dollar--euro, which accounted for $28\%$ of
global turnover, followed by the dollar--yen $(14\%)$, dollar--sterling
$(9\%)$, euro--yen $(3\%)$, euro--sterling $(3\%)$ and sterling--yen $(3\%)$
--- see \citet[p.~15]{BIS}. In the terminology of Section~\ref{S-d5} of this
paper, these four currencies may be called ``leading currencies.'' In a
world containing just these four leading currencies, arbitrage operations
can display periodicity, thus being bounded, or ``mild'' exponential growth.
Once we move to account for triangular arbitrage outside these four leading
currencies the arbitrage sequences may be unbounded, displaying
``explosive'' double exponential behaviour. In practice central banks often
intervene to attempt to prevent exchange rates exceeding certain bounds. The
history of financial crises, however, is replete with examples of central
bank interventions in FX markets being unsuccessful, the central banks being
eventually obliged to abandon their exchange rate pegs or target ranges.

The ``leading exchange rate indices,'' discussed at the end of
Section~\ref{S-d5} of this paper, indicate arbitrage profit opportunities
expressed in an FX trader's ``home currency'' unit of account. The August 15
1971 abandonment of the convertibility of the US dollar into gold at a fixed
price of 35\$ per ounce, and the breakdown of the Bretton Woods system
whereby IMF member currencies were pegged to the US dollar, left the
international monetary system without a well-defined numeraire or unit of
account. To fill this gap, currency-specific effective exchange rate (EER)
indices are calculated by national central banks and finance ministries, and
by international organisations such as the IMF and BIS. An EER is a
geometrically weighted average of the bilateral exchange rates for a
particular currency. So, for example, the BIS currently publishes broad EERs
covering $61$ currencies, and narrow EERs covering $27$ currencies, using
2008--2010 international trade flow data for weights
\citep{KlauFung06,BIS12}. The various EERs use the ``home currency'' unit of
account to provide exchange rate indices designed for a world in which there
is no gold or single reserve currency anchor.

Special Drawing Rights (SDRs) were created by the IMF in 1969 to supplement
the gold and foreign exchange reserves available to member central banks.
The SDR was initially defined as equivalent to $0.888671$ grams of fine
gold, which could be exchanged for $1$ US dollar in 1969, but is now defined
as a weighted value of the US dollar, euro, pound sterling and Japanese yen,
with the weights being changed every five years, expressed in US dollars
\citep{IMF12}. A European Currency Unit (ECU), conceived by European
Economic Community (EEC, later the EU) in 1979, served as an EU unit of
account, being a weighted value of EU member currencies, until the euro was
launched in 1999.

The governor of the People's Bank of China recently advocated that the SDR
be adopted as a global, tradeable reserve currency \citep{ZX09}. The SDR is
a dollar--euro--sterling--yen currency index, expressed in US dollars, so
this is akin to the ``leading exchange rate'' index faced by a dollar FX
trader-arbitrageur in the $d=4$ case considered earlier in the present
paper. Accordingly, the periodicity and exponential growth features of
arbitrage sequences involving these four currencies would apply. This
suggests that central banks holding the SDR as a tradeable reserve currency
would be faced with endemic, rather than self-extinguishing, opportunities
for arbitrage profits on their FX reserve portfolios. The proposal that
``the basket of currencies forming the basis for SDR valuation should be
expanded to include currencies of all major economics'' \citep[p.~3]{ZX09}
would exacerbate  the arbitrage problem, introducing the ``inheritance of
instability'' arising in the $d\ge 5$ case considered in the present paper.

\newpage\appendix \hfill{\Large \textbf{\appendixname}}

\section{Matrices $\boldsymbol{B_{(ijk)}, D_{(ijk)}}$
and $\boldsymbol{G_{(ijk)}}$ in the Case
$\boldsymbol{d=4}$}\label{A-d4}

The matrices $B_{\omega}$:
{%
\renewcommand{\arraystretch}{0.7}
\begin{alignat*}{2}
B_{(123)}&=\left(\begin{array}{@{}R@{}R@{}R@{}R@{}R@{}R@{}}
      0&1&0&-1&0&0\\
      0&1&0&0&0&0\\
      0&0&1&0&0&0\\
      0&0&0&1&0&0\\
      0&0&0&0&1&0\\
      0&0&0&0&0&1
\end{array}\right),&\quad
B_{(124)}&=\left(\begin{array}{@{}R@{}R@{}R@{}R@{}R@{}R@{}}
      0&0&1&0&-1&0\\
      0&1&0&0&0&0\\
      0&0&1&0&0&0\\
      0&0&0&1&0&0\\
      0&0&0&0&1&0\\
      0&0&0&0&0&1
\end{array}\right),\\
B_{(132)}&=\left(\begin{array}{@{}R@{}R@{}R@{}R@{}R@{}R@{}}
      \vphantom{-}1&0&0&0&0&0\\
      1&0&0&1&0&0\\
      0&0&1&0&0&0\\
      0&0&0&1&0&0\\
      0&0&0&0&1&0\\
      0&0&0&0&0&1
\end{array}\right),&\quad
B_{(134)}&=\left(\begin{array}{@{}R@{}R@{}R@{}R@{}R@{}R@{}}
      1&0&0&0&0&0\\
      0&0&1&0&0&-1\\
      0&0&1&0&0&0\\
      0&0&0&1&0&0\\
      0&0&0&0&1&0\\
      0&0&0&0&0&1
\end{array}\right),\\
B_{(142)}&=\left(\begin{array}{@{}R@{}R@{}R@{}R@{}R@{}R@{}}
      \vphantom{-}1&0&0&0&0&0\\
      0&1&0&0&0&0\\
      1&0&0&0&1&0\\
      0&0&0&1&0&0\\
      0&0&0&0&1&0\\
      0&0&0&0&0&1
\end{array}\right),&\quad
B_{(143)}&=\left(\begin{array}{@{}R@{}R@{}R@{}R@{}R@{}R@{}}
      \vphantom{-}1&0&0&0&0&0\\
      0&1&0&0&0&0\\
      0&1&0&0&0&1\\
      0&0&0&1&0&0\\
      0&0&0&0&1&0\\
      0&0&0&0&0&1
\end{array}\right),\\
B_{(231)}&=\left(\begin{array}{@{}R@{}R@{}R@{}R@{}R@{}R@{}}
      1&0&0&0&0&0\\
      0&1&0&0&0&0\\
      0&0&1&0&0&0\\
      -1&1&0&0&0&0\\
      0&0&0&0&1&0\\
      0&0&0&0&0&1
\end{array}\right),&\quad
B_{(234)}&=\left(\begin{array}{@{}R@{}R@{}R@{}R@{}R@{}R@{}}
      1&0&0&0&0&0\\
      0&1&0&0&0&0\\
      0&0&1&0&0&0\\
      0&0&0&0&1&-1\\
      0&0&0&0&1&0\\
      0&0&0&0&0&1
\end{array}\right),\\
B_{(241)}&=\left(\begin{array}{@{}R@{}R@{}R@{}R@{}R@{}R@{}}
      1&0&0&0&0&0\\
      0&1&0&0&0&0\\
      0&0&1&0&0&0\\
      0&0&0&1&0&0\\
      -1&0&1&0&0&0\\
      0&0&0&0&0&1
\end{array}\right),&\quad
B_{(243)}&=\left(\begin{array}{@{}R@{}R@{}R@{}R@{}R@{}R@{}}
      \vphantom{-}1&0&0&0&0&0\\
      0&1&0&0&0&0\\
      0&0&1&0&0&0\\
      0&0&0&1&0&0\\
      0&0&0&1&0&1\\
      0&0&0&0&0&1
\end{array}\right),\\
B_{(341)}&=\left(\begin{array}{@{}R@{}R@{}R@{}R@{}R@{}R@{}}
      1&0&0&0&0&0\\
      0&1&0&0&0&0\\
      0&0&1&0&0&0\\
      0&0&0&1&0&0\\
      0&0&0&0&1&0\\
      0&-1&1&0&0&0
\end{array}\right),&\quad
B_{(342)}&=\left(\begin{array}{@{}R@{}R@{}R@{}R@{}R@{}R@{}}
      1&0&0&0&0&0\\
      0&1&0&0&0&0\\
      0&0&1&0&0&0\\
      0&0&0&1&0&0\\
      0&0&0&0&1&0\\
      0&0&0&-1&1&0
\end{array}\right).
\end{alignat*}
}%

The matrices $Q$, $Q^{-1}$ and $D_{\omega}=Q^{-1}B_{\omega}Q$ take
the form:
{\renewcommand{\arraystretch}{0.7}%
\begin{alignat*}{2}
Q&=\left(\begin{array}{@{}R@{}R@{}R|@{~}R@{}R@{}R@{}}
1&0&0&0&0&0\\
0&1&0&0&0&0\\
0&0&1&0&0&0\\
\hline
-1&1&0&1&0&0\\
-1&0&1&0&1&0\\
0&-1&1&0&0&1
\end{array}\right),&\quad
Q^{-1}&=\left(\begin{array}{@{}R@{}R@{}R|@{~}R@{}R@{}R@{}}
1&0&0&0&0&0\\
0&1&0&0&0&0\\
0&0&1&0&0&0\\
\hline
1&-1&0&1&0&0\\
1&0&-1&0&1&0\\
0&1&-1&0&0&1
\end{array}\right),\\
D_{(123)}&=\left(\begin{array}{@{}R@{}R@{}R|@{~}R@{}R@{}R@{}}
1&0&0&-1&0&0\\
0&1&0&0&0&0\\
0&0&1&0&0&0\\
\hline
0&0&0&0&0&0\\
0&0&0&-1&1&0\\
0&0&0&0&0&1
\end{array}\right),&\quad
D_{(124)}&=\left(\begin{array}{@{}R@{}R@{}R|@{~}R@{}R@{}R@{}}
1&0&0&0&-1&0\\
0&1&0&0&0&0\\
0&0&1&0&0&0\\
\hline
0&0&0&1&-1&0\\
0&0&0&0&0&0\\
0&0&0&0&0&1
\end{array}\right),\\
D_{(132)}&=\left(\begin{array}{@{}R@{}R@{}R|@{~}R@{}R@{}R@{}}
1&0&0&0&0&0\\
0&1&0&1&0&0\\
0&0&1&0&0&0\\
\hline
0&0&0&0&0&0\\
0&0&0&0&1&0\\
0&0&0&1&0&1
\end{array}\right),&\quad
D_{(134)}&=\left(\begin{array}{@{}R@{}R@{}R|@{~}R@{}R@{}R@{}}
1&0&0&0&0&0\\
0&1&0&0&0&-1\\
0&0&1&0&0&0\\
\hline
0&0&0&1&0&1\\
0&0&0&0&1&0\\
0&0&0&0&0&0
\end{array}\right),\\
D_{(142)}&=\left(\begin{array}{@{}R@{}R@{}R|@{~}R@{}R@{}R@{}}
1&0&0&0&0&0\\
0&1&0&0&0&0\\
0&0&1&0&1&0\\
\hline
0&0&0&1&0&0\\
0&0&0&0&0&0\\
0&0&0&0&-1&1
\end{array}\right),&\quad
D_{(143)}&=\left(\begin{array}{@{}R@{}R@{}R|@{~}R@{}R@{}R@{}}
1&0&0&0&0&0\\
0&1&0&0&0&0\\
0&0&1&0&0&1\\
\hline
0&0&0&1&0&0\\
0&0&0&0&1&-1\\
0&0&0&0&0&0
\end{array}\right),\\
D_{(231)}&=\left(\begin{array}{@{}R@{}R@{}R|@{~}R@{}R@{}R@{}}
\vphantom{-}1&0&0&0&0&0\\
0&1&0&0&0&0\\
0&0&1&0&0&0\\
\hline
0&0&0&0&0&0\\
0&0&0&0&1&0\\
0&0&0&0&0&1
\end{array}\right),&\quad
D_{(234)}&=\left(\begin{array}{@{}R@{}R@{}R|@{~}R@{}R@{}R@{}}
1&0&0&0&0&0\\
0&1&0&0&0&0\\
0&0&1&0&0&0\\
\hline
0&0&0&0&1&-1\\
0&0&0&0&1&0\\
0&0&0&0&0&1
\end{array}\right),\\
D_{(241)}&=\left(\begin{array}{@{}R@{}R@{}R|@{~}R@{}R@{}R@{}}
\vphantom{-}1&0&0&0&0&0\\
0&1&0&0&0&0\\
0&0&1&0&0&0\\
\hline
0&0&0&1&0&0\\
0&0&0&0&0&0\\
0&0&0&0&0&1
\end{array}\right),&\quad
D_{(243)}&=\left(\begin{array}{@{}R@{}R@{}R|@{~}R@{}R@{}R@{}}
\vphantom{-}1&0&0&0&0&0\\
0&1&0&0&0&0\\
0&0&1&0&0&0\\
\hline
0&0&0&1&0&0\\
0&0&0&1&0&1\\
0&0&0&0&0&1
\end{array}\right),\\
D_{(341)}&=\left(\begin{array}{@{}R@{}R@{}R|@{~}R@{}R@{}R@{}}
\vphantom{-}1&0&0&0&0&0\\
0&1&0&0&0&0\\
0&0&1&0&0&0\\
\hline
0&0&0&1&0&0\\
0&0&0&0&1&0\\
0&0&0&0&0&0
\end{array}\right),&\quad
D_{(342)}&=\left(\begin{array}{@{}R@{}R@{}R|@{~}R@{}R@{}R@{}}
1&0&0&0&0&0\\
0&1&0&0&0&0\\
0&0&1&0&0&0\\
\hline
0&0&0&1&0&0\\
0&0&0&0&1&0\\
0&0&0&-1&1&0
\end{array}\right).
\end{alignat*}
}%

The matrices $G_{\omega}$:
{\renewcommand{\arraystretch}{0.7}%
\begin{alignat*}{3}
G_{(123)}&=\left(\begin{array}{@{}R@{}R@{}R@{}}
0&0&0\\
-1&1&0\\
0&0&1
\end{array}\right),&\quad
G_{(124)}&=\left(\begin{array}{@{}R@{}R@{}R@{}}
1&-1&0\\
0&0&0\\
0&0&1
\end{array}\right),&\quad
G_{(132)}&=\left(\begin{array}{@{}R@{}R@{}R@{}}
0&0&0\\
0&1&0\\
1&0&1
\end{array}\right),\\
G_{(134)}&=\left(\begin{array}{@{}R@{}R@{}R@{}}
1&0&1\\
0&1&0\\
0&0&0
\end{array}\right),&\quad
G_{(142)}&=\left(\begin{array}{@{}R@{}R@{}R@{}}
1&0&0\\
0&0&0\\
0&-1&1
\end{array}\right),&\quad
G_{(143)}&=\left(\begin{array}{@{}R@{}R@{}R@{}}
1&0&0\\
0&1&-1\\
0&0&0
\end{array}\right),\\
G_{(231)}&=\left(\begin{array}{@{}R@{}R@{}R@{}}
0&0&0\\
0&1&0\\
0&0&1
\end{array}\right),&\quad
G_{(234)}&=\left(\begin{array}{@{}R@{}R@{}R@{}}
0&1&-1\\
0&1&0\\
0&0&1
\end{array}\right),&\quad
G_{(241)}&=\left(\begin{array}{@{}R@{}R@{}R@{}}
1&0&0\\
0&0&0\\
0&0&1
\end{array}\right),\\
G_{(243)}&=\left(\begin{array}{@{}R@{}R@{}R@{}}
1&0&0\\
1&0&1\\
0&0&1
\end{array}\right),&\quad
G_{(341)}&=\left(\begin{array}{@{}R@{}R@{}R@{}}
1&0&0\\
0&1&0\\
0&0&0
\end{array}\right),&\quad
G_{(342)}&=\left(\begin{array}{@{}R@{}R@{}R@{}}
1&0&0\\
0&1&0\\
-1&1&0
\end{array}\right).
\end{alignat*}
}%

\section{Matrices $\boldsymbol{B_{(ijk)}, D_{(ijk)}}$
and $\boldsymbol{G_{(ijk)}}$ in the case
$\boldsymbol{d=5}$}\label{A-d5}

The matrices $B_{\omega}$:
{\renewcommand{\arraystretch}{0.7}%
\begin{alignat*}{2}
B_{(123)} &=
\left(\begin{array}{@{}R@{}R@{}R@{}R@{}R@{}R@{}R@{}R@{}R@{}R@{}}
0 & 1 & 0 & 0 & -1 & 0 & 0 & 0 & 0 & 0 \\
0 & 1 & 0 & 0 & 0 & 0 & 0 & 0 & 0 & 0 \\
0 & 0 & 1 & 0 & 0 & 0 & 0 & 0 & 0 & 0 \\
0 & 0 & 0 & 1 & 0 & 0 & 0 & 0 & 0 & 0 \\
0 & 0 & 0 & 0 & 1 & 0 & 0 & 0 & 0 & 0 \\
0 & 0 & 0 & 0 & 0 & 1 & 0 & 0 & 0 & 0 \\
0 & 0 & 0 & 0 & 0 & 0 & 1 & 0 & 0 & 0 \\
0 & 0 & 0 & 0 & 0 & 0 & 0 & 1 & 0 & 0 \\
0 & 0 & 0 & 0 & 0 & 0 & 0 & 0 & 1 & 0 \\
0 & 0 & 0 & 0 & 0 & 0 & 0 & 0 & 0 & 1
\end{array}\right),&\quad
B_{(124)} &=
\left(\begin{array}{@{}R@{}R@{}R@{}R@{}R@{}R@{}R@{}R@{}R@{}R@{}}
0 & 0 & 1 & 0 & 0 & -1 & 0 & 0 & 0 & 0 \\
0 & 1 & 0 & 0 & 0 & 0 & 0 & 0 & 0 & 0 \\
0 & 0 & 1 & 0 & 0 & 0 & 0 & 0 & 0 & 0 \\
0 & 0 & 0 & 1 & 0 & 0 & 0 & 0 & 0 & 0 \\
0 & 0 & 0 & 0 & 1 & 0 & 0 & 0 & 0 & 0 \\
0 & 0 & 0 & 0 & 0 & 1 & 0 & 0 & 0 & 0 \\
0 & 0 & 0 & 0 & 0 & 0 & 1 & 0 & 0 & 0 \\
0 & 0 & 0 & 0 & 0 & 0 & 0 & 1 & 0 & 0 \\
0 & 0 & 0 & 0 & 0 & 0 & 0 & 0 & 1 & 0 \\
0 & 0 & 0 & 0 & 0 & 0 & 0 & 0 & 0 & 1
\end{array}\right),\\
B_{(125)} &=
\left(\begin{array}{@{}R@{}R@{}R@{}R@{}R@{}R@{}R@{}R@{}R@{}R@{}}
0 & 0 & 0 & 1 & 0 & 0 & -1 & 0 & 0 & 0 \\
0 & 1 & 0 & 0 & 0 & 0 & 0 & 0 & 0 & 0 \\
0 & 0 & 1 & 0 & 0 & 0 & 0 & 0 & 0 & 0 \\
0 & 0 & 0 & 1 & 0 & 0 & 0 & 0 & 0 & 0 \\
0 & 0 & 0 & 0 & 1 & 0 & 0 & 0 & 0 & 0 \\
0 & 0 & 0 & 0 & 0 & 1 & 0 & 0 & 0 & 0 \\
0 & 0 & 0 & 0 & 0 & 0 & 1 & 0 & 0 & 0 \\
0 & 0 & 0 & 0 & 0 & 0 & 0 & 1 & 0 & 0 \\
0 & 0 & 0 & 0 & 0 & 0 & 0 & 0 & 1 & 0 \\
0 & 0 & 0 & 0 & 0 & 0 & 0 & 0 & 0 & 1
\end{array}\right),&\quad
B_{(132)} &=
\left(\begin{array}{@{}R@{}R@{}R@{}R@{}R@{}R@{}R@{}R@{}R@{}R@{}}
1 & 0 & 0 & 0 & 0 & 0 & 0 & 0 & 0 & 0 \\
1 & 0 & 0 & 0 & 1 & 0 & 0 & 0 & 0 & 0 \\
0 & 0 & 1 & 0 & 0 & 0 & 0 & 0 & 0 & 0 \\
0 & 0 & 0 & 1 & 0 & 0 & 0 & 0 & 0 & 0 \\
0 & 0 & 0 & 0 & 1 & 0 & 0 & 0 & 0 & 0 \\
0 & 0 & 0 & 0 & 0 & 1 & 0 & 0 & 0 & 0 \\
0 & 0 & 0 & 0 & 0 & 0 & 1 & 0 & 0 & 0 \\
0 & 0 & 0 & 0 & 0 & 0 & 0 & 1 & 0 & 0 \\
0 & 0 & 0 & 0 & 0 & 0 & 0 & 0 & 1 & 0 \\
0 & 0 & 0 & 0 & 0 & 0 & 0 & 0 & 0 & 1
\end{array}\right),\\
B_{(134)} &=
\left(\begin{array}{@{}R@{}R@{}R@{}R@{}R@{}R@{}R@{}R@{}R@{}R@{}}
1 & 0 & 0 & 0 & 0 & 0 & 0 & 0 & 0 & 0 \\
0 & 0 & 1 & 0 & 0 & 0 & 0 & -1 & 0 & 0 \\
0 & 0 & 1 & 0 & 0 & 0 & 0 & 0 & 0 & 0 \\
0 & 0 & 0 & 1 & 0 & 0 & 0 & 0 & 0 & 0 \\
0 & 0 & 0 & 0 & 1 & 0 & 0 & 0 & 0 & 0 \\
0 & 0 & 0 & 0 & 0 & 1 & 0 & 0 & 0 & 0 \\
0 & 0 & 0 & 0 & 0 & 0 & 1 & 0 & 0 & 0 \\
0 & 0 & 0 & 0 & 0 & 0 & 0 & 1 & 0 & 0 \\
0 & 0 & 0 & 0 & 0 & 0 & 0 & 0 & 1 & 0 \\
0 & 0 & 0 & 0 & 0 & 0 & 0 & 0 & 0 & 1
\end{array}\right),&\quad
B_{(135)} &=
\left(\begin{array}{@{}R@{}R@{}R@{}R@{}R@{}R@{}R@{}R@{}R@{}R@{}}
1 & 0 & 0 & 0 & 0 & 0 & 0 & 0 & 0 & 0 \\
0 & 0 & 0 & 1 & 0 & 0 & 0 & 0 & -1 & 0 \\
0 & 0 & 1 & 0 & 0 & 0 & 0 & 0 & 0 & 0 \\
0 & 0 & 0 & 1 & 0 & 0 & 0 & 0 & 0 & 0 \\
0 & 0 & 0 & 0 & 1 & 0 & 0 & 0 & 0 & 0 \\
0 & 0 & 0 & 0 & 0 & 1 & 0 & 0 & 0 & 0 \\
0 & 0 & 0 & 0 & 0 & 0 & 1 & 0 & 0 & 0 \\
0 & 0 & 0 & 0 & 0 & 0 & 0 & 1 & 0 & 0 \\
0 & 0 & 0 & 0 & 0 & 0 & 0 & 0 & 1 & 0 \\
0 & 0 & 0 & 0 & 0 & 0 & 0 & 0 & 0 & 1
\end{array}\right),\\
B_{(142)} &=
\left(\begin{array}{@{}R@{}R@{}R@{}R@{}R@{}R@{}R@{}R@{}R@{}R@{}}
1 & 0 & 0 & 0 & 0 & 0 & 0 & 0 & 0 & 0 \\
0 & 1 & 0 & 0 & 0 & 0 & 0 & 0 & 0 & 0 \\
1 & 0 & 0 & 0 & 0 & 1 & 0 & 0 & 0 & 0 \\
0 & 0 & 0 & 1 & 0 & 0 & 0 & 0 & 0 & 0 \\
0 & 0 & 0 & 0 & 1 & 0 & 0 & 0 & 0 & 0 \\
0 & 0 & 0 & 0 & 0 & 1 & 0 & 0 & 0 & 0 \\
0 & 0 & 0 & 0 & 0 & 0 & 1 & 0 & 0 & 0 \\
0 & 0 & 0 & 0 & 0 & 0 & 0 & 1 & 0 & 0 \\
0 & 0 & 0 & 0 & 0 & 0 & 0 & 0 & 1 & 0 \\
0 & 0 & 0 & 0 & 0 & 0 & 0 & 0 & 0 & 1
\end{array}\right),&\quad
B_{(143)} &=
\left(\begin{array}{@{}R@{}R@{}R@{}R@{}R@{}R@{}R@{}R@{}R@{}R@{}}
1 & 0 & 0 & 0 & 0 & 0 & 0 & 0 & 0 & 0 \\
0 & 1 & 0 & 0 & 0 & 0 & 0 & 0 & 0 & 0 \\
0 & 1 & 0 & 0 & 0 & 0 & 0 & 1 & 0 & 0 \\
0 & 0 & 0 & 1 & 0 & 0 & 0 & 0 & 0 & 0 \\
0 & 0 & 0 & 0 & 1 & 0 & 0 & 0 & 0 & 0 \\
0 & 0 & 0 & 0 & 0 & 1 & 0 & 0 & 0 & 0 \\
0 & 0 & 0 & 0 & 0 & 0 & 1 & 0 & 0 & 0 \\
0 & 0 & 0 & 0 & 0 & 0 & 0 & 1 & 0 & 0 \\
0 & 0 & 0 & 0 & 0 & 0 & 0 & 0 & 1 & 0 \\
0 & 0 & 0 & 0 & 0 & 0 & 0 & 0 & 0 & 1
\end{array}\right),\\
B_{(145)} &=
\left(\begin{array}{@{}R@{}R@{}R@{}R@{}R@{}R@{}R@{}R@{}R@{}R@{}}
1 & 0 & 0 & 0 & 0 & 0 & 0 & 0 & 0 & 0 \\
0 & 1 & 0 & 0 & 0 & 0 & 0 & 0 & 0 & 0 \\
0 & 0 & 0 & 1 & 0 & 0 & 0 & 0 & 0 & -1 \\
0 & 0 & 0 & 1 & 0 & 0 & 0 & 0 & 0 & 0 \\
0 & 0 & 0 & 0 & 1 & 0 & 0 & 0 & 0 & 0 \\
0 & 0 & 0 & 0 & 0 & 1 & 0 & 0 & 0 & 0 \\
0 & 0 & 0 & 0 & 0 & 0 & 1 & 0 & 0 & 0 \\
0 & 0 & 0 & 0 & 0 & 0 & 0 & 1 & 0 & 0 \\
0 & 0 & 0 & 0 & 0 & 0 & 0 & 0 & 1 & 0 \\
0 & 0 & 0 & 0 & 0 & 0 & 0 & 0 & 0 & 1
\end{array}\right),&\quad
B_{(152)} &=
\left(\begin{array}{@{}R@{}R@{}R@{}R@{}R@{}R@{}R@{}R@{}R@{}R@{}}
1 & 0 & 0 & 0 & 0 & 0 & 0 & 0 & 0 & 0 \\
0 & 1 & 0 & 0 & 0 & 0 & 0 & 0 & 0 & 0 \\
0 & 0 & 1 & 0 & 0 & 0 & 0 & 0 & 0 & 0 \\
1 & 0 & 0 & 0 & 0 & 0 & 1 & 0 & 0 & 0 \\
0 & 0 & 0 & 0 & 1 & 0 & 0 & 0 & 0 & 0 \\
0 & 0 & 0 & 0 & 0 & 1 & 0 & 0 & 0 & 0 \\
0 & 0 & 0 & 0 & 0 & 0 & 1 & 0 & 0 & 0 \\
0 & 0 & 0 & 0 & 0 & 0 & 0 & 1 & 0 & 0 \\
0 & 0 & 0 & 0 & 0 & 0 & 0 & 0 & 1 & 0 \\
0 & 0 & 0 & 0 & 0 & 0 & 0 & 0 & 0 & 1
\end{array}\right),\\
B_{(153)} &=
\left(\begin{array}{@{}R@{}R@{}R@{}R@{}R@{}R@{}R@{}R@{}R@{}R@{}}
1 & 0 & 0 & 0 & 0 & 0 & 0 & 0 & 0 & 0 \\
0 & 1 & 0 & 0 & 0 & 0 & 0 & 0 & 0 & 0 \\
0 & 0 & 1 & 0 & 0 & 0 & 0 & 0 & 0 & 0 \\
0 & 1 & 0 & 0 & 0 & 0 & 0 & 0 & 1 & 0 \\
0 & 0 & 0 & 0 & 1 & 0 & 0 & 0 & 0 & 0 \\
0 & 0 & 0 & 0 & 0 & 1 & 0 & 0 & 0 & 0 \\
0 & 0 & 0 & 0 & 0 & 0 & 1 & 0 & 0 & 0 \\
0 & 0 & 0 & 0 & 0 & 0 & 0 & 1 & 0 & 0 \\
0 & 0 & 0 & 0 & 0 & 0 & 0 & 0 & 1 & 0 \\
0 & 0 & 0 & 0 & 0 & 0 & 0 & 0 & 0 & 1
\end{array}\right),&\quad
B_{(154)} &=
\left(\begin{array}{@{}R@{}R@{}R@{}R@{}R@{}R@{}R@{}R@{}R@{}R@{}}
1 & 0 & 0 & 0 & 0 & 0 & 0 & 0 & 0 & 0 \\
0 & 1 & 0 & 0 & 0 & 0 & 0 & 0 & 0 & 0 \\
0 & 0 & 1 & 0 & 0 & 0 & 0 & 0 & 0 & 0 \\
0 & 0 & 1 & 0 & 0 & 0 & 0 & 0 & 0 & 1 \\
0 & 0 & 0 & 0 & 1 & 0 & 0 & 0 & 0 & 0 \\
0 & 0 & 0 & 0 & 0 & 1 & 0 & 0 & 0 & 0 \\
0 & 0 & 0 & 0 & 0 & 0 & 1 & 0 & 0 & 0 \\
0 & 0 & 0 & 0 & 0 & 0 & 0 & 1 & 0 & 0 \\
0 & 0 & 0 & 0 & 0 & 0 & 0 & 0 & 1 & 0 \\
0 & 0 & 0 & 0 & 0 & 0 & 0 & 0 & 0 & 1
\end{array}\right),\\
B_{(231)} &=
\left(\begin{array}{@{}R@{}R@{}R@{}R@{}R@{}R@{}R@{}R@{}R@{}R@{}}
1 & 0 & 0 & 0 & 0 & 0 & 0 & 0 & 0 & 0 \\
0 & 1 & 0 & 0 & 0 & 0 & 0 & 0 & 0 & 0 \\
0 & 0 & 1 & 0 & 0 & 0 & 0 & 0 & 0 & 0 \\
0 & 0 & 0 & 1 & 0 & 0 & 0 & 0 & 0 & 0 \\
-1 & 1 & 0 & 0 & 0 & 0 & 0 & 0 & 0 & 0 \\
0 & 0 & 0 & 0 & 0 & 1 & 0 & 0 & 0 & 0 \\
0 & 0 & 0 & 0 & 0 & 0 & 1 & 0 & 0 & 0 \\
0 & 0 & 0 & 0 & 0 & 0 & 0 & 1 & 0 & 0 \\
0 & 0 & 0 & 0 & 0 & 0 & 0 & 0 & 1 & 0 \\
0 & 0 & 0 & 0 & 0 & 0 & 0 & 0 & 0 & 1
\end{array}\right),&\quad
B_{(234)} &=
\left(\begin{array}{@{}R@{}R@{}R@{}R@{}R@{}R@{}R@{}R@{}R@{}R@{}}
1 & 0 & 0 & 0 & 0 & 0 & 0 & 0 & 0 & 0 \\
0 & 1 & 0 & 0 & 0 & 0 & 0 & 0 & 0 & 0 \\
0 & 0 & 1 & 0 & 0 & 0 & 0 & 0 & 0 & 0 \\
0 & 0 & 0 & 1 & 0 & 0 & 0 & 0 & 0 & 0 \\
0 & 0 & 0 & 0 & 0 & 1 & 0 & -1 & 0 & 0 \\
0 & 0 & 0 & 0 & 0 & 1 & 0 & 0 & 0 & 0 \\
0 & 0 & 0 & 0 & 0 & 0 & 1 & 0 & 0 & 0 \\
0 & 0 & 0 & 0 & 0 & 0 & 0 & 1 & 0 & 0 \\
0 & 0 & 0 & 0 & 0 & 0 & 0 & 0 & 1 & 0 \\
0 & 0 & 0 & 0 & 0 & 0 & 0 & 0 & 0 & 1
\end{array}\right),\\
B_{(235)} &=
\left(\begin{array}{@{}R@{}R@{}R@{}R@{}R@{}R@{}R@{}R@{}R@{}R@{}}
1 & 0 & 0 & 0 & 0 & 0 & 0 & 0 & 0 & 0 \\
0 & 1 & 0 & 0 & 0 & 0 & 0 & 0 & 0 & 0 \\
0 & 0 & 1 & 0 & 0 & 0 & 0 & 0 & 0 & 0 \\
0 & 0 & 0 & 1 & 0 & 0 & 0 & 0 & 0 & 0 \\
0 & 0 & 0 & 0 & 0 & 0 & 1 & 0 & -1 & 0 \\
0 & 0 & 0 & 0 & 0 & 1 & 0 & 0 & 0 & 0 \\
0 & 0 & 0 & 0 & 0 & 0 & 1 & 0 & 0 & 0 \\
0 & 0 & 0 & 0 & 0 & 0 & 0 & 1 & 0 & 0 \\
0 & 0 & 0 & 0 & 0 & 0 & 0 & 0 & 1 & 0 \\
0 & 0 & 0 & 0 & 0 & 0 & 0 & 0 & 0 & 1
\end{array}\right),&\quad
B_{(241)} &=
\left(\begin{array}{@{}R@{}R@{}R@{}R@{}R@{}R@{}R@{}R@{}R@{}R@{}}
1 & 0 & 0 & 0 & 0 & 0 & 0 & 0 & 0 & 0 \\
0 & 1 & 0 & 0 & 0 & 0 & 0 & 0 & 0 & 0 \\
0 & 0 & 1 & 0 & 0 & 0 & 0 & 0 & 0 & 0 \\
0 & 0 & 0 & 1 & 0 & 0 & 0 & 0 & 0 & 0 \\
0 & 0 & 0 & 0 & 1 & 0 & 0 & 0 & 0 & 0 \\
-1 & 0 & 1 & 0 & 0 & 0 & 0 & 0 & 0 & 0 \\
0 & 0 & 0 & 0 & 0 & 0 & 1 & 0 & 0 & 0 \\
0 & 0 & 0 & 0 & 0 & 0 & 0 & 1 & 0 & 0 \\
0 & 0 & 0 & 0 & 0 & 0 & 0 & 0 & 1 & 0 \\
0 & 0 & 0 & 0 & 0 & 0 & 0 & 0 & 0 & 1
\end{array}\right),\\
B_{(243)} &=
\left(\begin{array}{@{}R@{}R@{}R@{}R@{}R@{}R@{}R@{}R@{}R@{}R@{}}
1 & 0 & 0 & 0 & 0 & 0 & 0 & 0 & 0 & 0 \\
0 & 1 & 0 & 0 & 0 & 0 & 0 & 0 & 0 & 0 \\
0 & 0 & 1 & 0 & 0 & 0 & 0 & 0 & 0 & 0 \\
0 & 0 & 0 & 1 & 0 & 0 & 0 & 0 & 0 & 0 \\
0 & 0 & 0 & 0 & 1 & 0 & 0 & 0 & 0 & 0 \\
0 & 0 & 0 & 0 & 1 & 0 & 0 & 1 & 0 & 0 \\
0 & 0 & 0 & 0 & 0 & 0 & 1 & 0 & 0 & 0 \\
0 & 0 & 0 & 0 & 0 & 0 & 0 & 1 & 0 & 0 \\
0 & 0 & 0 & 0 & 0 & 0 & 0 & 0 & 1 & 0 \\
0 & 0 & 0 & 0 & 0 & 0 & 0 & 0 & 0 & 1
\end{array}\right),&\quad
B_{(245)} &=
\left(\begin{array}{@{}R@{}R@{}R@{}R@{}R@{}R@{}R@{}R@{}R@{}R@{}}
1 & 0 & 0 & 0 & 0 & 0 & 0 & 0 & 0 & 0 \\
0 & 1 & 0 & 0 & 0 & 0 & 0 & 0 & 0 & 0 \\
0 & 0 & 1 & 0 & 0 & 0 & 0 & 0 & 0 & 0 \\
0 & 0 & 0 & 1 & 0 & 0 & 0 & 0 & 0 & 0 \\
0 & 0 & 0 & 0 & 1 & 0 & 0 & 0 & 0 & 0 \\
0 & 0 & 0 & 0 & 0 & 0 & 1 & 0 & 0 & -1 \\
0 & 0 & 0 & 0 & 0 & 0 & 1 & 0 & 0 & 0 \\
0 & 0 & 0 & 0 & 0 & 0 & 0 & 1 & 0 & 0 \\
0 & 0 & 0 & 0 & 0 & 0 & 0 & 0 & 1 & 0 \\
0 & 0 & 0 & 0 & 0 & 0 & 0 & 0 & 0 & 1
\end{array}\right),\\
B_{(251)} &=
\left(\begin{array}{@{}R@{}R@{}R@{}R@{}R@{}R@{}R@{}R@{}R@{}R@{}}
1 & 0 & 0 & 0 & 0 & 0 & 0 & 0 & 0 & 0 \\
0 & 1 & 0 & 0 & 0 & 0 & 0 & 0 & 0 & 0 \\
0 & 0 & 1 & 0 & 0 & 0 & 0 & 0 & 0 & 0 \\
0 & 0 & 0 & 1 & 0 & 0 & 0 & 0 & 0 & 0 \\
0 & 0 & 0 & 0 & 1 & 0 & 0 & 0 & 0 & 0 \\
0 & 0 & 0 & 0 & 0 & 1 & 0 & 0 & 0 & 0 \\
-1 & 0 & 0 & 1 & 0 & 0 & 0 & 0 & 0 & 0 \\
0 & 0 & 0 & 0 & 0 & 0 & 0 & 1 & 0 & 0 \\
0 & 0 & 0 & 0 & 0 & 0 & 0 & 0 & 1 & 0 \\
0 & 0 & 0 & 0 & 0 & 0 & 0 & 0 & 0 & 1
\end{array}\right),&\quad
B_{(253)} &=
\left(\begin{array}{@{}R@{}R@{}R@{}R@{}R@{}R@{}R@{}R@{}R@{}R@{}}
1 & 0 & 0 & 0 & 0 & 0 & 0 & 0 & 0 & 0 \\
0 & 1 & 0 & 0 & 0 & 0 & 0 & 0 & 0 & 0 \\
0 & 0 & 1 & 0 & 0 & 0 & 0 & 0 & 0 & 0 \\
0 & 0 & 0 & 1 & 0 & 0 & 0 & 0 & 0 & 0 \\
0 & 0 & 0 & 0 & 1 & 0 & 0 & 0 & 0 & 0 \\
0 & 0 & 0 & 0 & 0 & 1 & 0 & 0 & 0 & 0 \\
0 & 0 & 0 & 0 & 1 & 0 & 0 & 0 & 1 & 0 \\
0 & 0 & 0 & 0 & 0 & 0 & 0 & 1 & 0 & 0 \\
0 & 0 & 0 & 0 & 0 & 0 & 0 & 0 & 1 & 0 \\
0 & 0 & 0 & 0 & 0 & 0 & 0 & 0 & 0 & 1
\end{array}\right),\\
B_{(254)} &=
\left(\begin{array}{@{}R@{}R@{}R@{}R@{}R@{}R@{}R@{}R@{}R@{}R@{}}
1 & 0 & 0 & 0 & 0 & 0 & 0 & 0 & 0 & 0 \\
0 & 1 & 0 & 0 & 0 & 0 & 0 & 0 & 0 & 0 \\
0 & 0 & 1 & 0 & 0 & 0 & 0 & 0 & 0 & 0 \\
0 & 0 & 0 & 1 & 0 & 0 & 0 & 0 & 0 & 0 \\
0 & 0 & 0 & 0 & 1 & 0 & 0 & 0 & 0 & 0 \\
0 & 0 & 0 & 0 & 0 & 1 & 0 & 0 & 0 & 0 \\
0 & 0 & 0 & 0 & 0 & 1 & 0 & 0 & 0 & 1 \\
0 & 0 & 0 & 0 & 0 & 0 & 0 & 1 & 0 & 0 \\
0 & 0 & 0 & 0 & 0 & 0 & 0 & 0 & 1 & 0 \\
0 & 0 & 0 & 0 & 0 & 0 & 0 & 0 & 0 & 1
\end{array}\right),&\quad
B_{(341)} &=
\left(\begin{array}{@{}R@{}R@{}R@{}R@{}R@{}R@{}R@{}R@{}R@{}R@{}}
1 & 0 & 0 & 0 & 0 & 0 & 0 & 0 & 0 & 0 \\
0 & 1 & 0 & 0 & 0 & 0 & 0 & 0 & 0 & 0 \\
0 & 0 & 1 & 0 & 0 & 0 & 0 & 0 & 0 & 0 \\
0 & 0 & 0 & 1 & 0 & 0 & 0 & 0 & 0 & 0 \\
0 & 0 & 0 & 0 & 1 & 0 & 0 & 0 & 0 & 0 \\
0 & 0 & 0 & 0 & 0 & 1 & 0 & 0 & 0 & 0 \\
0 & 0 & 0 & 0 & 0 & 0 & 1 & 0 & 0 & 0 \\
0 & -1 & 1 & 0 & 0 & 0 & 0 & 0 & 0 & 0 \\
0 & 0 & 0 & 0 & 0 & 0 & 0 & 0 & 1 & 0 \\
0 & 0 & 0 & 0 & 0 & 0 & 0 & 0 & 0 & 1
\end{array}\right),\\
B_{(342)} &=
\left(\begin{array}{@{}R@{}R@{}R@{}R@{}R@{}R@{}R@{}R@{}R@{}R@{}}
1 & 0 & 0 & 0 & 0 & 0 & 0 & 0 & 0 & 0 \\
0 & 1 & 0 & 0 & 0 & 0 & 0 & 0 & 0 & 0 \\
0 & 0 & 1 & 0 & 0 & 0 & 0 & 0 & 0 & 0 \\
0 & 0 & 0 & 1 & 0 & 0 & 0 & 0 & 0 & 0 \\
0 & 0 & 0 & 0 & 1 & 0 & 0 & 0 & 0 & 0 \\
0 & 0 & 0 & 0 & 0 & 1 & 0 & 0 & 0 & 0 \\
0 & 0 & 0 & 0 & 0 & 0 & 1 & 0 & 0 & 0 \\
0 & 0 & 0 & 0 & -1 & 1 & 0 & 0 & 0 & 0 \\
0 & 0 & 0 & 0 & 0 & 0 & 0 & 0 & 1 & 0 \\
0 & 0 & 0 & 0 & 0 & 0 & 0 & 0 & 0 & 1
\end{array}\right),&\quad
B_{(345)} &=
\left(\begin{array}{@{}R@{}R@{}R@{}R@{}R@{}R@{}R@{}R@{}R@{}R@{}}
1 & 0 & 0 & 0 & 0 & 0 & 0 & 0 & 0 & 0 \\
0 & 1 & 0 & 0 & 0 & 0 & 0 & 0 & 0 & 0 \\
0 & 0 & 1 & 0 & 0 & 0 & 0 & 0 & 0 & 0 \\
0 & 0 & 0 & 1 & 0 & 0 & 0 & 0 & 0 & 0 \\
0 & 0 & 0 & 0 & 1 & 0 & 0 & 0 & 0 & 0 \\
0 & 0 & 0 & 0 & 0 & 1 & 0 & 0 & 0 & 0 \\
0 & 0 & 0 & 0 & 0 & 0 & 1 & 0 & 0 & 0 \\
0 & 0 & 0 & 0 & 0 & 0 & 0 & 0 & 1 & -1 \\
0 & 0 & 0 & 0 & 0 & 0 & 0 & 0 & 1 & 0 \\
0 & 0 & 0 & 0 & 0 & 0 & 0 & 0 & 0 & 1
\end{array}\right),\\
B_{(351)} &=
\left(\begin{array}{@{}R@{}R@{}R@{}R@{}R@{}R@{}R@{}R@{}R@{}R@{}}
1 & 0 & 0 & 0 & 0 & 0 & 0 & 0 & 0 & 0 \\
0 & 1 & 0 & 0 & 0 & 0 & 0 & 0 & 0 & 0 \\
0 & 0 & 1 & 0 & 0 & 0 & 0 & 0 & 0 & 0 \\
0 & 0 & 0 & 1 & 0 & 0 & 0 & 0 & 0 & 0 \\
0 & 0 & 0 & 0 & 1 & 0 & 0 & 0 & 0 & 0 \\
0 & 0 & 0 & 0 & 0 & 1 & 0 & 0 & 0 & 0 \\
0 & 0 & 0 & 0 & 0 & 0 & 1 & 0 & 0 & 0 \\
0 & 0 & 0 & 0 & 0 & 0 & 0 & 1 & 0 & 0 \\
0 & -1 & 0 & 1 & 0 & 0 & 0 & 0 & 0 & 0 \\
0 & 0 & 0 & 0 & 0 & 0 & 0 & 0 & 0 & 1
\end{array}\right),&\quad
B_{(352)} &=
\left(\begin{array}{@{}R@{}R@{}R@{}R@{}R@{}R@{}R@{}R@{}R@{}R@{}}
1 & 0 & 0 & 0 & 0 & 0 & 0 & 0 & 0 & 0 \\
0 & 1 & 0 & 0 & 0 & 0 & 0 & 0 & 0 & 0 \\
0 & 0 & 1 & 0 & 0 & 0 & 0 & 0 & 0 & 0 \\
0 & 0 & 0 & 1 & 0 & 0 & 0 & 0 & 0 & 0 \\
0 & 0 & 0 & 0 & 1 & 0 & 0 & 0 & 0 & 0 \\
0 & 0 & 0 & 0 & 0 & 1 & 0 & 0 & 0 & 0 \\
0 & 0 & 0 & 0 & 0 & 0 & 1 & 0 & 0 & 0 \\
0 & 0 & 0 & 0 & 0 & 0 & 0 & 1 & 0 & 0 \\
0 & 0 & 0 & 0 & -1 & 0 & 1 & 0 & 0 & 0 \\
0 & 0 & 0 & 0 & 0 & 0 & 0 & 0 & 0 & 1
\end{array}\right),\\
B_{(354)} &=
\left(\begin{array}{@{}R@{}R@{}R@{}R@{}R@{}R@{}R@{}R@{}R@{}R@{}}
1 & 0 & 0 & 0 & 0 & 0 & 0 & 0 & 0 & 0 \\
0 & 1 & 0 & 0 & 0 & 0 & 0 & 0 & 0 & 0 \\
0 & 0 & 1 & 0 & 0 & 0 & 0 & 0 & 0 & 0 \\
0 & 0 & 0 & 1 & 0 & 0 & 0 & 0 & 0 & 0 \\
0 & 0 & 0 & 0 & 1 & 0 & 0 & 0 & 0 & 0 \\
0 & 0 & 0 & 0 & 0 & 1 & 0 & 0 & 0 & 0 \\
0 & 0 & 0 & 0 & 0 & 0 & 1 & 0 & 0 & 0 \\
0 & 0 & 0 & 0 & 0 & 0 & 0 & 1 & 0 & 0 \\
0 & 0 & 0 & 0 & 0 & 0 & 0 & 1 & 0 & 1 \\
0 & 0 & 0 & 0 & 0 & 0 & 0 & 0 & 0 & 1
\end{array}\right),&\quad
B_{(451)} &=
\left(\begin{array}{@{}R@{}R@{}R@{}R@{}R@{}R@{}R@{}R@{}R@{}R@{}}
1 & 0 & 0 & 0 & 0 & 0 & 0 & 0 & 0 & 0 \\
0 & 1 & 0 & 0 & 0 & 0 & 0 & 0 & 0 & 0 \\
0 & 0 & 1 & 0 & 0 & 0 & 0 & 0 & 0 & 0 \\
0 & 0 & 0 & 1 & 0 & 0 & 0 & 0 & 0 & 0 \\
0 & 0 & 0 & 0 & 1 & 0 & 0 & 0 & 0 & 0 \\
0 & 0 & 0 & 0 & 0 & 1 & 0 & 0 & 0 & 0 \\
0 & 0 & 0 & 0 & 0 & 0 & 1 & 0 & 0 & 0 \\
0 & 0 & 0 & 0 & 0 & 0 & 0 & 1 & 0 & 0 \\
0 & 0 & 0 & 0 & 0 & 0 & 0 & 0 & 1 & 0 \\
0 & 0 & -1 & 1 & 0 & 0 & 0 & 0 & 0 & 0
\end{array}\right),\\
B_{(452)} &=
\left(\begin{array}{@{}R@{}R@{}R@{}R@{}R@{}R@{}R@{}R@{}R@{}R@{}}
1 & 0 & 0 & 0 & 0 & 0 & 0 & 0 & 0 & 0 \\
0 & 1 & 0 & 0 & 0 & 0 & 0 & 0 & 0 & 0 \\
0 & 0 & 1 & 0 & 0 & 0 & 0 & 0 & 0 & 0 \\
0 & 0 & 0 & 1 & 0 & 0 & 0 & 0 & 0 & 0 \\
0 & 0 & 0 & 0 & 1 & 0 & 0 & 0 & 0 & 0 \\
0 & 0 & 0 & 0 & 0 & 1 & 0 & 0 & 0 & 0 \\
0 & 0 & 0 & 0 & 0 & 0 & 1 & 0 & 0 & 0 \\
0 & 0 & 0 & 0 & 0 & 0 & 0 & 1 & 0 & 0 \\
0 & 0 & 0 & 0 & 0 & 0 & 0 & 0 & 1 & 0 \\
0 & 0 & 0 & 0 & 0 & -1 & 1 & 0 & 0 & 0
\end{array}\right),&\quad
B_{(453)} &=
\left(\begin{array}{@{}R@{}R@{}R@{}R@{}R@{}R@{}R@{}R@{}R@{}R@{}}
1 & 0 & 0 & 0 & 0 & 0 & 0 & 0 & 0 & 0 \\
0 & 1 & 0 & 0 & 0 & 0 & 0 & 0 & 0 & 0 \\
0 & 0 & 1 & 0 & 0 & 0 & 0 & 0 & 0 & 0 \\
0 & 0 & 0 & 1 & 0 & 0 & 0 & 0 & 0 & 0 \\
0 & 0 & 0 & 0 & 1 & 0 & 0 & 0 & 0 & 0 \\
0 & 0 & 0 & 0 & 0 & 1 & 0 & 0 & 0 & 0 \\
0 & 0 & 0 & 0 & 0 & 0 & 1 & 0 & 0 & 0 \\
0 & 0 & 0 & 0 & 0 & 0 & 0 & 1 & 0 & 0 \\
0 & 0 & 0 & 0 & 0 & 0 & 0 & 0 & 1 & 0 \\
0 & 0 & 0 & 0 & 0 & 0 & 0 & -1 & 1 & 0
\end{array}\right).
\end{alignat*}
}%

The matrices $Q$, $Q^{-1}$ and $D_{\omega}=Q^{-1}B_{\omega}Q$ take
the form:
{\renewcommand{\arraystretch}{0.7}%
\begin{alignat*}{2}
Q &=
\left(\begin{array}{@{}R@{}R@{}R@{}R|@{~}R@{}R@{}R@{}R@{}R@{}R@{}}
1 & 0 & 0 & 0 & 0 & 0 & 0 & 0 & 0 & 0 \\
0 & 1 & 0 & 0 & 0 & 0 & 0 & 0 & 0 & 0 \\
0 & 0 & 1 & 0 & 0 & 0 & 0 & 0 & 0 & 0 \\
0 & 0 & 0 & 1 & 0 & 0 & 0 & 0 & 0 & 0 \\
\hline
-1 & 1 & 0 & 0 & 1 & 0 & 0 & 0 & 0 & 0 \\
-1 & 0 & 1 & 0 & 0 & 1 & 0 & 0 & 0 & 0 \\
-1 & 0 & 0 & 1 & 0 & 0 & 1 & 0 & 0 & 0 \\
0 & -1 & 1 & 0 & 0 & 0 & 0 & 1 & 0 & 0 \\
0 & -1 & 0 & 1 & 0 & 0 & 0 & 0 & 1 & 0 \\
0 & 0 & -1 & 1 & 0 & 0 & 0 & 0 & 0 & 1
\end{array}\right),&\quad
Q^{-1} &=
\left(\begin{array}{@{}R@{}R@{}R@{}R|@{~}R@{}R@{}R@{}R@{}R@{}R@{}}
1 & 0 & 0 & 0 & 0 & 0 & 0 & 0 & 0 & 0 \\
0 & 1 & 0 & 0 & 0 & 0 & 0 & 0 & 0 & 0 \\
0 & 0 & 1 & 0 & 0 & 0 & 0 & 0 & 0 & 0 \\
0 & 0 & 0 & 1 & 0 & 0 & 0 & 0 & 0 & 0 \\
\hline
1 & -1 & 0 & 0 & 1 & 0 & 0 & 0 & 0 & 0 \\
1 & 0 & -1 & 0 & 0 & 1 & 0 & 0 & 0 & 0 \\
1 & 0 & 0 & -1 & 0 & 0 & 1 & 0 & 0 & 0 \\
0 & 1 & -1 & 0 & 0 & 0 & 0 & 1 & 0 & 0 \\
0 & 1 & 0 & -1 & 0 & 0 & 0 & 0 & 1 & 0 \\
0 & 0 & 1 & -1 & 0 & 0 & 0 & 0 & 0 & 1
\end{array}\right),\\
D_{(123)} &=
\left(\begin{array}{@{}R@{}R@{}R@{}R|@{~}R@{}R@{}R@{}R@{}R@{}R@{}}
1 & 0 & 0 & 0 & -1 & 0 & 0 & 0 & 0 & 0 \\
0 & 1 & 0 & 0 & 0 & 0 & 0 & 0 & 0 & 0 \\
0 & 0 & 1 & 0 & 0 & 0 & 0 & 0 & 0 & 0 \\
0 & 0 & 0 & 1 & 0 & 0 & 0 & 0 & 0 & 0 \\
\hline
0 & 0 & 0 & 0 & 0 & 0 & 0 & 0 & 0 & 0 \\
0 & 0 & 0 & 0 & -1 & 1 & 0 & 0 & 0 & 0 \\
0 & 0 & 0 & 0 & -1 & 0 & 1 & 0 & 0 & 0 \\
0 & 0 & 0 & 0 & 0 & 0 & 0 & 1 & 0 & 0 \\
0 & 0 & 0 & 0 & 0 & 0 & 0 & 0 & 1 & 0 \\
0 & 0 & 0 & 0 & 0 & 0 & 0 & 0 & 0 & 1
\end{array}\right),&\quad
D_{(124)} &=
\left(\begin{array}{@{}R@{}R@{}R@{}R|@{~}R@{}R@{}R@{}R@{}R@{}R@{}}
1 & 0 & 0 & 0 & 0 & -1 & 0 & 0 & 0 & 0 \\
0 & 1 & 0 & 0 & 0 & 0 & 0 & 0 & 0 & 0 \\
0 & 0 & 1 & 0 & 0 & 0 & 0 & 0 & 0 & 0 \\
0 & 0 & 0 & 1 & 0 & 0 & 0 & 0 & 0 & 0 \\
\hline
0 & 0 & 0 & 0 & 1 & -1 & 0 & 0 & 0 & 0 \\
0 & 0 & 0 & 0 & 0 & 0 & 0 & 0 & 0 & 0 \\
0 & 0 & 0 & 0 & 0 & -1 & 1 & 0 & 0 & 0 \\
0 & 0 & 0 & 0 & 0 & 0 & 0 & 1 & 0 & 0 \\
0 & 0 & 0 & 0 & 0 & 0 & 0 & 0 & 1 & 0 \\
0 & 0 & 0 & 0 & 0 & 0 & 0 & 0 & 0 & 1
\end{array}\right),\\
D_{(125)} &=
\left(\begin{array}{@{}R@{}R@{}R@{}R|@{~}R@{}R@{}R@{}R@{}R@{}R@{}}
1 & 0 & 0 & 0 & 0 & 0 & -1 & 0 & 0 & 0 \\
0 & 1 & 0 & 0 & 0 & 0 & 0 & 0 & 0 & 0 \\
0 & 0 & 1 & 0 & 0 & 0 & 0 & 0 & 0 & 0 \\
0 & 0 & 0 & 1 & 0 & 0 & 0 & 0 & 0 & 0 \\
\hline
0 & 0 & 0 & 0 & 1 & 0 & -1 & 0 & 0 & 0 \\
0 & 0 & 0 & 0 & 0 & 1 & -1 & 0 & 0 & 0 \\
0 & 0 & 0 & 0 & 0 & 0 & 0 & 0 & 0 & 0 \\
0 & 0 & 0 & 0 & 0 & 0 & 0 & 1 & 0 & 0 \\
0 & 0 & 0 & 0 & 0 & 0 & 0 & 0 & 1 & 0 \\
0 & 0 & 0 & 0 & 0 & 0 & 0 & 0 & 0 & 1
\end{array}\right),&\quad
D_{(132)} &=
\left(\begin{array}{@{}R@{}R@{}R@{}R|@{~}R@{}R@{}R@{}R@{}R@{}R@{}}
1 & 0 & 0 & 0 & 0 & 0 & 0 & 0 & 0 & 0 \\
0 & 1 & 0 & 0 & 1 & 0 & 0 & 0 & 0 & 0 \\
0 & 0 & 1 & 0 & 0 & 0 & 0 & 0 & 0 & 0 \\
0 & 0 & 0 & 1 & 0 & 0 & 0 & 0 & 0 & 0 \\
\hline
0 & 0 & 0 & 0 & 0 & 0 & 0 & 0 & 0 & 0 \\
0 & 0 & 0 & 0 & 0 & 1 & 0 & 0 & 0 & 0 \\
0 & 0 & 0 & 0 & 0 & 0 & 1 & 0 & 0 & 0 \\
0 & 0 & 0 & 0 & 1 & 0 & 0 & 1 & 0 & 0 \\
0 & 0 & 0 & 0 & 1 & 0 & 0 & 0 & 1 & 0 \\
0 & 0 & 0 & 0 & 0 & 0 & 0 & 0 & 0 & 1
\end{array}\right),\\
D_{(134)} &=
\left(\begin{array}{@{}R@{}R@{}R@{}R|@{~}R@{}R@{}R@{}R@{}R@{}R@{}}
1 & 0 & 0 & 0 & 0 & 0 & 0 & 0 & 0 & 0 \\
0 & 1 & 0 & 0 & 0 & 0 & 0 & -1 & 0 & 0 \\
0 & 0 & 1 & 0 & 0 & 0 & 0 & 0 & 0 & 0 \\
0 & 0 & 0 & 1 & 0 & 0 & 0 & 0 & 0 & 0 \\
\hline
0 & 0 & 0 & 0 & 1 & 0 & 0 & 1 & 0 & 0 \\
0 & 0 & 0 & 0 & 0 & 1 & 0 & 0 & 0 & 0 \\
0 & 0 & 0 & 0 & 0 & 0 & 1 & 0 & 0 & 0 \\
0 & 0 & 0 & 0 & 0 & 0 & 0 & 0 & 0 & 0 \\
0 & 0 & 0 & 0 & 0 & 0 & 0 & -1 & 1 & 0 \\
0 & 0 & 0 & 0 & 0 & 0 & 0 & 0 & 0 & 1
\end{array}\right),&\quad
D_{(135)} &=
\left(\begin{array}{@{}R@{}R@{}R@{}R|@{~}R@{}R@{}R@{}R@{}R@{}R@{}}
1 & 0 & 0 & 0 & 0 & 0 & 0 & 0 & 0 & 0 \\
0 & 1 & 0 & 0 & 0 & 0 & 0 & 0 & -1 & 0 \\
0 & 0 & 1 & 0 & 0 & 0 & 0 & 0 & 0 & 0 \\
0 & 0 & 0 & 1 & 0 & 0 & 0 & 0 & 0 & 0 \\
\hline
0 & 0 & 0 & 0 & 1 & 0 & 0 & 0 & 1 & 0 \\
0 & 0 & 0 & 0 & 0 & 1 & 0 & 0 & 0 & 0 \\
0 & 0 & 0 & 0 & 0 & 0 & 1 & 0 & 0 & 0 \\
0 & 0 & 0 & 0 & 0 & 0 & 0 & 1 & -1 & 0 \\
0 & 0 & 0 & 0 & 0 & 0 & 0 & 0 & 0 & 0 \\
0 & 0 & 0 & 0 & 0 & 0 & 0 & 0 & 0 & 1
\end{array}\right),\\
D_{(142)} &=
\left(\begin{array}{@{}R@{}R@{}R@{}R|@{~}R@{}R@{}R@{}R@{}R@{}R@{}}
1 & 0 & 0 & 0 & 0 & 0 & 0 & 0 & 0 & 0 \\
0 & 1 & 0 & 0 & 0 & 0 & 0 & 0 & 0 & 0 \\
0 & 0 & 1 & 0 & 0 & 1 & 0 & 0 & 0 & 0 \\
0 & 0 & 0 & 1 & 0 & 0 & 0 & 0 & 0 & 0 \\
\hline
0 & 0 & 0 & 0 & 1 & 0 & 0 & 0 & 0 & 0 \\
0 & 0 & 0 & 0 & 0 & 0 & 0 & 0 & 0 & 0 \\
0 & 0 & 0 & 0 & 0 & 0 & 1 & 0 & 0 & 0 \\
0 & 0 & 0 & 0 & 0 & -1 & 0 & 1 & 0 & 0 \\
0 & 0 & 0 & 0 & 0 & 0 & 0 & 0 & 1 & 0 \\
0 & 0 & 0 & 0 & 0 & 1 & 0 & 0 & 0 & 1
\end{array}\right),&\quad
D_{(143)} &=
\left(\begin{array}{@{}R@{}R@{}R@{}R|@{~}R@{}R@{}R@{}R@{}R@{}R@{}}
1 & 0 & 0 & 0 & 0 & 0 & 0 & 0 & 0 & 0 \\
0 & 1 & 0 & 0 & 0 & 0 & 0 & 0 & 0 & 0 \\
0 & 0 & 1 & 0 & 0 & 0 & 0 & 1 & 0 & 0 \\
0 & 0 & 0 & 1 & 0 & 0 & 0 & 0 & 0 & 0 \\
\hline
0 & 0 & 0 & 0 & 1 & 0 & 0 & 0 & 0 & 0 \\
0 & 0 & 0 & 0 & 0 & 1 & 0 & -1 & 0 & 0 \\
0 & 0 & 0 & 0 & 0 & 0 & 1 & 0 & 0 & 0 \\
0 & 0 & 0 & 0 & 0 & 0 & 0 & 0 & 0 & 0 \\
0 & 0 & 0 & 0 & 0 & 0 & 0 & 0 & 1 & 0 \\
0 & 0 & 0 & 0 & 0 & 0 & 0 & 1 & 0 & 1
\end{array}\right),\\
D_{(145)} &=
\left(\begin{array}{@{}R@{}R@{}R@{}R|@{~}R@{}R@{}R@{}R@{}R@{}R@{}}
1 & 0 & 0 & 0 & 0 & 0 & 0 & 0 & 0 & 0 \\
0 & 1 & 0 & 0 & 0 & 0 & 0 & 0 & 0 & 0 \\
0 & 0 & 1 & 0 & 0 & 0 & 0 & 0 & 0 & -1 \\
0 & 0 & 0 & 1 & 0 & 0 & 0 & 0 & 0 & 0 \\
\hline
0 & 0 & 0 & 0 & 1 & 0 & 0 & 0 & 0 & 0 \\
0 & 0 & 0 & 0 & 0 & 1 & 0 & 0 & 0 & 1 \\
0 & 0 & 0 & 0 & 0 & 0 & 1 & 0 & 0 & 0 \\
0 & 0 & 0 & 0 & 0 & 0 & 0 & 1 & 0 & 1 \\
0 & 0 & 0 & 0 & 0 & 0 & 0 & 0 & 1 & 0 \\
0 & 0 & 0 & 0 & 0 & 0 & 0 & 0 & 0 & 0
\end{array}\right),&\quad
D_{(152)} &=
\left(\begin{array}{@{}R@{}R@{}R@{}R|@{~}R@{}R@{}R@{}R@{}R@{}R@{}}
1 & 0 & 0 & 0 & 0 & 0 & 0 & 0 & 0 & 0 \\
0 & 1 & 0 & 0 & 0 & 0 & 0 & 0 & 0 & 0 \\
0 & 0 & 1 & 0 & 0 & 0 & 0 & 0 & 0 & 0 \\
0 & 0 & 0 & 1 & 0 & 0 & 1 & 0 & 0 & 0 \\
\hline
0 & 0 & 0 & 0 & 1 & 0 & 0 & 0 & 0 & 0 \\
0 & 0 & 0 & 0 & 0 & 1 & 0 & 0 & 0 & 0 \\
0 & 0 & 0 & 0 & 0 & 0 & 0 & 0 & 0 & 0 \\
0 & 0 & 0 & 0 & 0 & 0 & 0 & 1 & 0 & 0 \\
0 & 0 & 0 & 0 & 0 & 0 & -1 & 0 & 1 & 0 \\
0 & 0 & 0 & 0 & 0 & 0 & -1 & 0 & 0 & 1
\end{array}\right),\\
D_{(153)} &=
\left(\begin{array}{@{}R@{}R@{}R@{}R|@{~}R@{}R@{}R@{}R@{}R@{}R@{}}
1 & 0 & 0 & 0 & 0 & 0 & 0 & 0 & 0 & 0 \\
0 & 1 & 0 & 0 & 0 & 0 & 0 & 0 & 0 & 0 \\
0 & 0 & 1 & 0 & 0 & 0 & 0 & 0 & 0 & 0 \\
0 & 0 & 0 & 1 & 0 & 0 & 0 & 0 & 1 & 0 \\
\hline
0 & 0 & 0 & 0 & 1 & 0 & 0 & 0 & 0 & 0 \\
0 & 0 & 0 & 0 & 0 & 1 & 0 & 0 & 0 & 0 \\
0 & 0 & 0 & 0 & 0 & 0 & 1 & 0 & -1 & 0 \\
0 & 0 & 0 & 0 & 0 & 0 & 0 & 1 & 0 & 0 \\
0 & 0 & 0 & 0 & 0 & 0 & 0 & 0 & 0 & 0 \\
0 & 0 & 0 & 0 & 0 & 0 & 0 & 0 & -1 & 1
\end{array}\right),&\quad
D_{(154)} &=
\left(\begin{array}{@{}R@{}R@{}R@{}R|@{~}R@{}R@{}R@{}R@{}R@{}R@{}}
1 & 0 & 0 & 0 & 0 & 0 & 0 & 0 & 0 & 0 \\
0 & 1 & 0 & 0 & 0 & 0 & 0 & 0 & 0 & 0 \\
0 & 0 & 1 & 0 & 0 & 0 & 0 & 0 & 0 & 0 \\
0 & 0 & 0 & 1 & 0 & 0 & 0 & 0 & 0 & 1 \\
\hline
0 & 0 & 0 & 0 & 1 & 0 & 0 & 0 & 0 & 0 \\
0 & 0 & 0 & 0 & 0 & 1 & 0 & 0 & 0 & 0 \\
0 & 0 & 0 & 0 & 0 & 0 & 1 & 0 & 0 & -1 \\
0 & 0 & 0 & 0 & 0 & 0 & 0 & 1 & 0 & 0 \\
0 & 0 & 0 & 0 & 0 & 0 & 0 & 0 & 1 & -1 \\
0 & 0 & 0 & 0 & 0 & 0 & 0 & 0 & 0 & 0
\end{array}\right),\\
D_{(231)} &=
\left(\begin{array}{@{}R@{}R@{}R@{}R|@{~}R@{}R@{}R@{}R@{}R@{}R@{}}
1 & 0 & 0 & 0 & 0 & 0 & 0 & 0 & 0 & 0 \\
0 & 1 & 0 & 0 & 0 & 0 & 0 & 0 & 0 & 0 \\
0 & 0 & 1 & 0 & 0 & 0 & 0 & 0 & 0 & 0 \\
0 & 0 & 0 & 1 & 0 & 0 & 0 & 0 & 0 & 0 \\
\hline
0 & 0 & 0 & 0 & 0 & 0 & 0 & 0 & 0 & 0 \\
0 & 0 & 0 & 0 & 0 & 1 & 0 & 0 & 0 & 0 \\
0 & 0 & 0 & 0 & 0 & 0 & 1 & 0 & 0 & 0 \\
0 & 0 & 0 & 0 & 0 & 0 & 0 & 1 & 0 & 0 \\
0 & 0 & 0 & 0 & 0 & 0 & 0 & 0 & 1 & 0 \\
0 & 0 & 0 & 0 & 0 & 0 & 0 & 0 & 0 & 1
\end{array}\right),&\quad
D_{(234)} &=
\left(\begin{array}{@{}R@{}R@{}R@{}R|@{~}R@{}R@{}R@{}R@{}R@{}R@{}}
1 & 0 & 0 & 0 & 0 & 0 & 0 & 0 & 0 & 0 \\
0 & 1 & 0 & 0 & 0 & 0 & 0 & 0 & 0 & 0 \\
0 & 0 & 1 & 0 & 0 & 0 & 0 & 0 & 0 & 0 \\
0 & 0 & 0 & 1 & 0 & 0 & 0 & 0 & 0 & 0 \\
\hline
0 & 0 & 0 & 0 & 0 & 1 & 0 & -1 & 0 & 0 \\
0 & 0 & 0 & 0 & 0 & 1 & 0 & 0 & 0 & 0 \\
0 & 0 & 0 & 0 & 0 & 0 & 1 & 0 & 0 & 0 \\
0 & 0 & 0 & 0 & 0 & 0 & 0 & 1 & 0 & 0 \\
0 & 0 & 0 & 0 & 0 & 0 & 0 & 0 & 1 & 0 \\
0 & 0 & 0 & 0 & 0 & 0 & 0 & 0 & 0 & 1
\end{array}\right),\\
D_{(235)} &=
\left(\begin{array}{@{}R@{}R@{}R@{}R|@{~}R@{}R@{}R@{}R@{}R@{}R@{}}
1 & 0 & 0 & 0 & 0 & 0 & 0 & 0 & 0 & 0 \\
0 & 1 & 0 & 0 & 0 & 0 & 0 & 0 & 0 & 0 \\
0 & 0 & 1 & 0 & 0 & 0 & 0 & 0 & 0 & 0 \\
0 & 0 & 0 & 1 & 0 & 0 & 0 & 0 & 0 & 0 \\
\hline
0 & 0 & 0 & 0 & 0 & 0 & 1 & 0 & -1 & 0 \\
0 & 0 & 0 & 0 & 0 & 1 & 0 & 0 & 0 & 0 \\
0 & 0 & 0 & 0 & 0 & 0 & 1 & 0 & 0 & 0 \\
0 & 0 & 0 & 0 & 0 & 0 & 0 & 1 & 0 & 0 \\
0 & 0 & 0 & 0 & 0 & 0 & 0 & 0 & 1 & 0 \\
0 & 0 & 0 & 0 & 0 & 0 & 0 & 0 & 0 & 1
\end{array}\right),&\quad
D_{(241)} &=
\left(\begin{array}{@{}R@{}R@{}R@{}R|@{~}R@{}R@{}R@{}R@{}R@{}R@{}}
1 & 0 & 0 & 0 & 0 & 0 & 0 & 0 & 0 & 0 \\
0 & 1 & 0 & 0 & 0 & 0 & 0 & 0 & 0 & 0 \\
0 & 0 & 1 & 0 & 0 & 0 & 0 & 0 & 0 & 0 \\
0 & 0 & 0 & 1 & 0 & 0 & 0 & 0 & 0 & 0 \\
\hline
0 & 0 & 0 & 0 & 1 & 0 & 0 & 0 & 0 & 0 \\
0 & 0 & 0 & 0 & 0 & 0 & 0 & 0 & 0 & 0 \\
0 & 0 & 0 & 0 & 0 & 0 & 1 & 0 & 0 & 0 \\
0 & 0 & 0 & 0 & 0 & 0 & 0 & 1 & 0 & 0 \\
0 & 0 & 0 & 0 & 0 & 0 & 0 & 0 & 1 & 0 \\
0 & 0 & 0 & 0 & 0 & 0 & 0 & 0 & 0 & 1
\end{array}\right),\\
D_{(243)} &=
\left(\begin{array}{@{}R@{}R@{}R@{}R|@{~}R@{}R@{}R@{}R@{}R@{}R@{}}
1 & 0 & 0 & 0 & 0 & 0 & 0 & 0 & 0 & 0 \\
0 & 1 & 0 & 0 & 0 & 0 & 0 & 0 & 0 & 0 \\
0 & 0 & 1 & 0 & 0 & 0 & 0 & 0 & 0 & 0 \\
0 & 0 & 0 & 1 & 0 & 0 & 0 & 0 & 0 & 0 \\
\hline
0 & 0 & 0 & 0 & 1 & 0 & 0 & 0 & 0 & 0 \\
0 & 0 & 0 & 0 & 1 & 0 & 0 & 1 & 0 & 0 \\
0 & 0 & 0 & 0 & 0 & 0 & 1 & 0 & 0 & 0 \\
0 & 0 & 0 & 0 & 0 & 0 & 0 & 1 & 0 & 0 \\
0 & 0 & 0 & 0 & 0 & 0 & 0 & 0 & 1 & 0 \\
0 & 0 & 0 & 0 & 0 & 0 & 0 & 0 & 0 & 1
\end{array}\right),&\quad
D_{(245)} &=
\left(\begin{array}{@{}R@{}R@{}R@{}R|@{~}R@{}R@{}R@{}R@{}R@{}R@{}}
1 & 0 & 0 & 0 & 0 & 0 & 0 & 0 & 0 & 0 \\
0 & 1 & 0 & 0 & 0 & 0 & 0 & 0 & 0 & 0 \\
0 & 0 & 1 & 0 & 0 & 0 & 0 & 0 & 0 & 0 \\
0 & 0 & 0 & 1 & 0 & 0 & 0 & 0 & 0 & 0 \\
\hline
0 & 0 & 0 & 0 & 1 & 0 & 0 & 0 & 0 & 0 \\
0 & 0 & 0 & 0 & 0 & 0 & 1 & 0 & 0 & -1 \\
0 & 0 & 0 & 0 & 0 & 0 & 1 & 0 & 0 & 0 \\
0 & 0 & 0 & 0 & 0 & 0 & 0 & 1 & 0 & 0 \\
0 & 0 & 0 & 0 & 0 & 0 & 0 & 0 & 1 & 0 \\
0 & 0 & 0 & 0 & 0 & 0 & 0 & 0 & 0 & 1
\end{array}\right),\\
D_{(251)} &=
\left(\begin{array}{@{}R@{}R@{}R@{}R|@{~}R@{}R@{}R@{}R@{}R@{}R@{}}
1 & 0 & 0 & 0 & 0 & 0 & 0 & 0 & 0 & 0 \\
0 & 1 & 0 & 0 & 0 & 0 & 0 & 0 & 0 & 0 \\
0 & 0 & 1 & 0 & 0 & 0 & 0 & 0 & 0 & 0 \\
0 & 0 & 0 & 1 & 0 & 0 & 0 & 0 & 0 & 0 \\
\hline
0 & 0 & 0 & 0 & 1 & 0 & 0 & 0 & 0 & 0 \\
0 & 0 & 0 & 0 & 0 & 1 & 0 & 0 & 0 & 0 \\
0 & 0 & 0 & 0 & 0 & 0 & 0 & 0 & 0 & 0 \\
0 & 0 & 0 & 0 & 0 & 0 & 0 & 1 & 0 & 0 \\
0 & 0 & 0 & 0 & 0 & 0 & 0 & 0 & 1 & 0 \\
0 & 0 & 0 & 0 & 0 & 0 & 0 & 0 & 0 & 1
\end{array}\right),&\quad
D_{(253)} &=
\left(\begin{array}{@{}R@{}R@{}R@{}R|@{~}R@{}R@{}R@{}R@{}R@{}R@{}}
1 & 0 & 0 & 0 & 0 & 0 & 0 & 0 & 0 & 0 \\
0 & 1 & 0 & 0 & 0 & 0 & 0 & 0 & 0 & 0 \\
0 & 0 & 1 & 0 & 0 & 0 & 0 & 0 & 0 & 0 \\
0 & 0 & 0 & 1 & 0 & 0 & 0 & 0 & 0 & 0 \\
\hline
0 & 0 & 0 & 0 & 1 & 0 & 0 & 0 & 0 & 0 \\
0 & 0 & 0 & 0 & 0 & 1 & 0 & 0 & 0 & 0 \\
0 & 0 & 0 & 0 & 1 & 0 & 0 & 0 & 1 & 0 \\
0 & 0 & 0 & 0 & 0 & 0 & 0 & 1 & 0 & 0 \\
0 & 0 & 0 & 0 & 0 & 0 & 0 & 0 & 1 & 0 \\
0 & 0 & 0 & 0 & 0 & 0 & 0 & 0 & 0 & 1
\end{array}\right),\\
D_{(254)} &=
\left(\begin{array}{@{}R@{}R@{}R@{}R|@{~}R@{}R@{}R@{}R@{}R@{}R@{}}
1 & 0 & 0 & 0 & 0 & 0 & 0 & 0 & 0 & 0 \\
0 & 1 & 0 & 0 & 0 & 0 & 0 & 0 & 0 & 0 \\
0 & 0 & 1 & 0 & 0 & 0 & 0 & 0 & 0 & 0 \\
0 & 0 & 0 & 1 & 0 & 0 & 0 & 0 & 0 & 0 \\
\hline
0 & 0 & 0 & 0 & 1 & 0 & 0 & 0 & 0 & 0 \\
0 & 0 & 0 & 0 & 0 & 1 & 0 & 0 & 0 & 0 \\
0 & 0 & 0 & 0 & 0 & 1 & 0 & 0 & 0 & 1 \\
0 & 0 & 0 & 0 & 0 & 0 & 0 & 1 & 0 & 0 \\
0 & 0 & 0 & 0 & 0 & 0 & 0 & 0 & 1 & 0 \\
0 & 0 & 0 & 0 & 0 & 0 & 0 & 0 & 0 & 1
\end{array}\right),&\quad
D_{(341)} &=
\left(\begin{array}{@{}R@{}R@{}R@{}R|@{~}R@{}R@{}R@{}R@{}R@{}R@{}}
1 & 0 & 0 & 0 & 0 & 0 & 0 & 0 & 0 & 0 \\
0 & 1 & 0 & 0 & 0 & 0 & 0 & 0 & 0 & 0 \\
0 & 0 & 1 & 0 & 0 & 0 & 0 & 0 & 0 & 0 \\
0 & 0 & 0 & 1 & 0 & 0 & 0 & 0 & 0 & 0 \\
\hline
0 & 0 & 0 & 0 & 1 & 0 & 0 & 0 & 0 & 0 \\
0 & 0 & 0 & 0 & 0 & 1 & 0 & 0 & 0 & 0 \\
0 & 0 & 0 & 0 & 0 & 0 & 1 & 0 & 0 & 0 \\
0 & 0 & 0 & 0 & 0 & 0 & 0 & 0 & 0 & 0 \\
0 & 0 & 0 & 0 & 0 & 0 & 0 & 0 & 1 & 0 \\
0 & 0 & 0 & 0 & 0 & 0 & 0 & 0 & 0 & 1
\end{array}\right),\\
D_{(342)} &=
\left(\begin{array}{@{}R@{}R@{}R@{}R|@{~}R@{}R@{}R@{}R@{}R@{}R@{}}
1 & 0 & 0 & 0 & 0 & 0 & 0 & 0 & 0 & 0 \\
0 & 1 & 0 & 0 & 0 & 0 & 0 & 0 & 0 & 0 \\
0 & 0 & 1 & 0 & 0 & 0 & 0 & 0 & 0 & 0 \\
0 & 0 & 0 & 1 & 0 & 0 & 0 & 0 & 0 & 0 \\
\hline
0 & 0 & 0 & 0 & 1 & 0 & 0 & 0 & 0 & 0 \\
0 & 0 & 0 & 0 & 0 & 1 & 0 & 0 & 0 & 0 \\
0 & 0 & 0 & 0 & 0 & 0 & 1 & 0 & 0 & 0 \\
0 & 0 & 0 & 0 & -1 & 1 & 0 & 0 & 0 & 0 \\
0 & 0 & 0 & 0 & 0 & 0 & 0 & 0 & 1 & 0 \\
0 & 0 & 0 & 0 & 0 & 0 & 0 & 0 & 0 & 1
\end{array}\right),&\quad
D_{(345)} &=
\left(\begin{array}{@{}R@{}R@{}R@{}R|@{~}R@{}R@{}R@{}R@{}R@{}R@{}}
1 & 0 & 0 & 0 & 0 & 0 & 0 & 0 & 0 & 0 \\
0 & 1 & 0 & 0 & 0 & 0 & 0 & 0 & 0 & 0 \\
0 & 0 & 1 & 0 & 0 & 0 & 0 & 0 & 0 & 0 \\
0 & 0 & 0 & 1 & 0 & 0 & 0 & 0 & 0 & 0 \\
\hline
0 & 0 & 0 & 0 & 1 & 0 & 0 & 0 & 0 & 0 \\
0 & 0 & 0 & 0 & 0 & 1 & 0 & 0 & 0 & 0 \\
0 & 0 & 0 & 0 & 0 & 0 & 1 & 0 & 0 & 0 \\
0 & 0 & 0 & 0 & 0 & 0 & 0 & 0 & 1 & -1 \\
0 & 0 & 0 & 0 & 0 & 0 & 0 & 0 & 1 & 0 \\
0 & 0 & 0 & 0 & 0 & 0 & 0 & 0 & 0 & 1
\end{array}\right),\\
D_{(351)} &=
\left(\begin{array}{@{}R@{}R@{}R@{}R|@{~}R@{}R@{}R@{}R@{}R@{}R@{}}
1 & 0 & 0 & 0 & 0 & 0 & 0 & 0 & 0 & 0 \\
0 & 1 & 0 & 0 & 0 & 0 & 0 & 0 & 0 & 0 \\
0 & 0 & 1 & 0 & 0 & 0 & 0 & 0 & 0 & 0 \\
0 & 0 & 0 & 1 & 0 & 0 & 0 & 0 & 0 & 0 \\
\hline
0 & 0 & 0 & 0 & 1 & 0 & 0 & 0 & 0 & 0 \\
0 & 0 & 0 & 0 & 0 & 1 & 0 & 0 & 0 & 0 \\
0 & 0 & 0 & 0 & 0 & 0 & 1 & 0 & 0 & 0 \\
0 & 0 & 0 & 0 & 0 & 0 & 0 & 1 & 0 & 0 \\
0 & 0 & 0 & 0 & 0 & 0 & 0 & 0 & 0 & 0 \\
0 & 0 & 0 & 0 & 0 & 0 & 0 & 0 & 0 & 1
\end{array}\right),&\quad
D_{(352)} &=
\left(\begin{array}{@{}R@{}R@{}R@{}R|@{~}R@{}R@{}R@{}R@{}R@{}R@{}}
1 & 0 & 0 & 0 & 0 & 0 & 0 & 0 & 0 & 0 \\
0 & 1 & 0 & 0 & 0 & 0 & 0 & 0 & 0 & 0 \\
0 & 0 & 1 & 0 & 0 & 0 & 0 & 0 & 0 & 0 \\
0 & 0 & 0 & 1 & 0 & 0 & 0 & 0 & 0 & 0 \\
\hline
0 & 0 & 0 & 0 & 1 & 0 & 0 & 0 & 0 & 0 \\
0 & 0 & 0 & 0 & 0 & 1 & 0 & 0 & 0 & 0 \\
0 & 0 & 0 & 0 & 0 & 0 & 1 & 0 & 0 & 0 \\
0 & 0 & 0 & 0 & 0 & 0 & 0 & 1 & 0 & 0 \\
0 & 0 & 0 & 0 & -1 & 0 & 1 & 0 & 0 & 0 \\
0 & 0 & 0 & 0 & 0 & 0 & 0 & 0 & 0 & 1
\end{array}\right),\\
D_{(354)} &=
\left(\begin{array}{@{}R@{}R@{}R@{}R|@{~}R@{}R@{}R@{}R@{}R@{}R@{}}
1 & 0 & 0 & 0 & 0 & 0 & 0 & 0 & 0 & 0 \\
0 & 1 & 0 & 0 & 0 & 0 & 0 & 0 & 0 & 0 \\
0 & 0 & 1 & 0 & 0 & 0 & 0 & 0 & 0 & 0 \\
0 & 0 & 0 & 1 & 0 & 0 & 0 & 0 & 0 & 0 \\
\hline
0 & 0 & 0 & 0 & 1 & 0 & 0 & 0 & 0 & 0 \\
0 & 0 & 0 & 0 & 0 & 1 & 0 & 0 & 0 & 0 \\
0 & 0 & 0 & 0 & 0 & 0 & 1 & 0 & 0 & 0 \\
0 & 0 & 0 & 0 & 0 & 0 & 0 & 1 & 0 & 0 \\
0 & 0 & 0 & 0 & 0 & 0 & 0 & 1 & 0 & 1 \\
0 & 0 & 0 & 0 & 0 & 0 & 0 & 0 & 0 & 1
\end{array}\right),&\quad
D_{(451)} &=
\left(\begin{array}{@{}R@{}R@{}R@{}R|@{~}R@{}R@{}R@{}R@{}R@{}R@{}}
1 & 0 & 0 & 0 & 0 & 0 & 0 & 0 & 0 & 0 \\
0 & 1 & 0 & 0 & 0 & 0 & 0 & 0 & 0 & 0 \\
0 & 0 & 1 & 0 & 0 & 0 & 0 & 0 & 0 & 0 \\
0 & 0 & 0 & 1 & 0 & 0 & 0 & 0 & 0 & 0 \\
\hline
0 & 0 & 0 & 0 & 1 & 0 & 0 & 0 & 0 & 0 \\
0 & 0 & 0 & 0 & 0 & 1 & 0 & 0 & 0 & 0 \\
0 & 0 & 0 & 0 & 0 & 0 & 1 & 0 & 0 & 0 \\
0 & 0 & 0 & 0 & 0 & 0 & 0 & 1 & 0 & 0 \\
0 & 0 & 0 & 0 & 0 & 0 & 0 & 0 & 1 & 0 \\
0 & 0 & 0 & 0 & 0 & 0 & 0 & 0 & 0 & 0
\end{array}\right),\\
D_{(452)} &=
\left(\begin{array}{@{}R@{}R@{}R@{}R|@{~}R@{}R@{}R@{}R@{}R@{}R@{}}
1 & 0 & 0 & 0 & 0 & 0 & 0 & 0 & 0 & 0 \\
0 & 1 & 0 & 0 & 0 & 0 & 0 & 0 & 0 & 0 \\
0 & 0 & 1 & 0 & 0 & 0 & 0 & 0 & 0 & 0 \\
0 & 0 & 0 & 1 & 0 & 0 & 0 & 0 & 0 & 0 \\
\hline
0 & 0 & 0 & 0 & 1 & 0 & 0 & 0 & 0 & 0 \\
0 & 0 & 0 & 0 & 0 & 1 & 0 & 0 & 0 & 0 \\
0 & 0 & 0 & 0 & 0 & 0 & 1 & 0 & 0 & 0 \\
0 & 0 & 0 & 0 & 0 & 0 & 0 & 1 & 0 & 0 \\
0 & 0 & 0 & 0 & 0 & 0 & 0 & 0 & 1 & 0 \\
0 & 0 & 0 & 0 & 0 & -1 & 1 & 0 & 0 & 0
\end{array}\right),&\quad
D_{(453)} &=
\left(\begin{array}{@{}R@{}R@{}R@{}R|@{~}R@{}R@{}R@{}R@{}R@{}R@{}}
1 & 0 & 0 & 0 & 0 & 0 & 0 & 0 & 0 & 0 \\
0 & 1 & 0 & 0 & 0 & 0 & 0 & 0 & 0 & 0 \\
0 & 0 & 1 & 0 & 0 & 0 & 0 & 0 & 0 & 0 \\
0 & 0 & 0 & 1 & 0 & 0 & 0 & 0 & 0 & 0 \\
\hline
0 & 0 & 0 & 0 & 1 & 0 & 0 & 0 & 0 & 0 \\
0 & 0 & 0 & 0 & 0 & 1 & 0 & 0 & 0 & 0 \\
0 & 0 & 0 & 0 & 0 & 0 & 1 & 0 & 0 & 0 \\
0 & 0 & 0 & 0 & 0 & 0 & 0 & 1 & 0 & 0 \\
0 & 0 & 0 & 0 & 0 & 0 & 0 & 0 & 1 & 0 \\
0 & 0 & 0 & 0 & 0 & 0 & 0 & -1 & 1 & 0
\end{array}\right).
\end{alignat*}
}%

The matrices $G_{\omega}$:
{\renewcommand{\arraystretch}{0.7}%
\begin{alignat*}{2}
G_{(123)} &= \left(\begin{array}{@{}R@{}R@{}R@{}R@{}R@{}R@{}}
0 & 0 & 0 & 0 & 0 & 0 \\
-1 & 1 & 0 & 0 & 0 & 0 \\
-1 & 0 & 1 & 0 & 0 & 0 \\
0 & 0 & 0 & 1 & 0 & 0 \\
0 & 0 & 0 & 0 & 1 & 0 \\
0 & 0 & 0 & 0 & 0 & 1
\end{array}\right),&\quad
G_{(124)} &= \left(\begin{array}{@{}R@{}R@{}R@{}R@{}R@{}R@{}}
1 & -1 & 0 & 0 & 0 & 0 \\
0 & 0 & 0 & 0 & 0 & 0 \\
0 & -1 & 1 & 0 & 0 & 0 \\
0 & 0 & 0 & 1 & 0 & 0 \\
0 & 0 & 0 & 0 & 1 & 0 \\
0 & 0 & 0 & 0 & 0 & 1
\end{array}\right),\\
G_{(125)} &= \left(\begin{array}{@{}R@{}R@{}R@{}R@{}R@{}R@{}}
1 & 0 & -1 & 0 & 0 & 0 \\
0 & 1 & -1 & 0 & 0 & 0 \\
0 & 0 & 0 & 0 & 0 & 0 \\
0 & 0 & 0 & 1 & 0 & 0 \\
0 & 0 & 0 & 0 & 1 & 0 \\
0 & 0 & 0 & 0 & 0 & 1
\end{array}\right),&\quad
G_{(132)} &= \left(\begin{array}{@{}R@{}R@{}R@{}R@{}R@{}R@{}}
0 & 0 & 0 & 0 & 0 & 0 \\
0 & 1 & 0 & 0 & 0 & 0 \\
0 & 0 & 1 & 0 & 0 & 0 \\
1 & 0 & 0 & 1 & 0 & 0 \\
1 & 0 & 0 & 0 & 1 & 0 \\
0 & 0 & 0 & 0 & 0 & 1
\end{array}\right),\\
G_{(134)} &= \left(\begin{array}{@{}R@{}R@{}R@{}R@{}R@{}R@{}}
1 & 0 & 0 & 1 & 0 & 0 \\
0 & 1 & 0 & 0 & 0 & 0 \\
0 & 0 & 1 & 0 & 0 & 0 \\
0 & 0 & 0 & 0 & 0 & 0 \\
0 & 0 & 0 & -1 & 1 & 0 \\
0 & 0 & 0 & 0 & 0 & 1
\end{array}\right),&\quad
G_{(135)} &= \left(\begin{array}{@{}R@{}R@{}R@{}R@{}R@{}R@{}}
1 & 0 & 0 & 0 & 1 & 0 \\
0 & 1 & 0 & 0 & 0 & 0 \\
0 & 0 & 1 & 0 & 0 & 0 \\
0 & 0 & 0 & 1 & -1 & 0 \\
0 & 0 & 0 & 0 & 0 & 0 \\
0 & 0 & 0 & 0 & 0 & 1
\end{array}\right),\\
G_{(142)} &= \left(\begin{array}{@{}R@{}R@{}R@{}R@{}R@{}R@{}}
1 & 0 & 0 & 0 & 0 & 0 \\
0 & 0 & 0 & 0 & 0 & 0 \\
0 & 0 & 1 & 0 & 0 & 0 \\
0 & -1 & 0 & 1 & 0 & 0 \\
0 & 0 & 0 & 0 & 1 & 0 \\
0 & 1 & 0 & 0 & 0 & 1
\end{array}\right),&\quad
G_{(143)} &= \left(\begin{array}{@{}R@{}R@{}R@{}R@{}R@{}R@{}}
1 & 0 & 0 & 0 & 0 & 0 \\
0 & 1 & 0 & -1 & 0 & 0 \\
0 & 0 & 1 & 0 & 0 & 0 \\
0 & 0 & 0 & 0 & 0 & 0 \\
0 & 0 & 0 & 0 & 1 & 0 \\
0 & 0 & 0 & 1 & 0 & 1
\end{array}\right),\\
G_{(145)} &= \left(\begin{array}{@{}R@{}R@{}R@{}R@{}R@{}R@{}}
1 & 0 & 0 & 0 & 0 & 0 \\
0 & 1 & 0 & 0 & 0 & 1 \\
0 & 0 & 1 & 0 & 0 & 0 \\
0 & 0 & 0 & 1 & 0 & 1 \\
0 & 0 & 0 & 0 & 1 & 0 \\
0 & 0 & 0 & 0 & 0 & 0
\end{array}\right),&\quad
G_{(152)} &= \left(\begin{array}{@{}R@{}R@{}R@{}R@{}R@{}R@{}}
1 & 0 & 0 & 0 & 0 & 0 \\
0 & 1 & 0 & 0 & 0 & 0 \\
0 & 0 & 0 & 0 & 0 & 0 \\
0 & 0 & 0 & 1 & 0 & 0 \\
0 & 0 & -1 & 0 & 1 & 0 \\
0 & 0 & -1 & 0 & 0 & 1
\end{array}\right),\\
G_{(153)} &= \left(\begin{array}{@{}R@{}R@{}R@{}R@{}R@{}R@{}}
1 & 0 & 0 & 0 & 0 & 0 \\
0 & 1 & 0 & 0 & 0 & 0 \\
0 & 0 & 1 & 0 & -1 & 0 \\
0 & 0 & 0 & 1 & 0 & 0 \\
0 & 0 & 0 & 0 & 0 & 0 \\
0 & 0 & 0 & 0 & -1 & 1
\end{array}\right),&\quad
G_{(154)} &= \left(\begin{array}{@{}R@{}R@{}R@{}R@{}R@{}R@{}}
1 & 0 & 0 & 0 & 0 & 0 \\
0 & 1 & 0 & 0 & 0 & 0 \\
0 & 0 & 1 & 0 & 0 & -1 \\
0 & 0 & 0 & 1 & 0 & 0 \\
0 & 0 & 0 & 0 & 1 & -1 \\
0 & 0 & 0 & 0 & 0 & 0
\end{array}\right),\\
G_{(231)} &= \left(\begin{array}{@{}R@{}R@{}R@{}R@{}R@{}R@{}}
0 & 0 & 0 & 0 & 0 & 0 \\
0 & 1 & 0 & 0 & 0 & 0 \\
0 & 0 & 1 & 0 & 0 & 0 \\
0 & 0 & 0 & 1 & 0 & 0 \\
0 & 0 & 0 & 0 & 1 & 0 \\
0 & 0 & 0 & 0 & 0 & 1
\end{array}\right),&\quad
G_{(234)} &= \left(\begin{array}{@{}R@{}R@{}R@{}R@{}R@{}R@{}}
0 & 1 & 0 & -1 & 0 & 0 \\
0 & 1 & 0 & 0 & 0 & 0 \\
0 & 0 & 1 & 0 & 0 & 0 \\
0 & 0 & 0 & 1 & 0 & 0 \\
0 & 0 & 0 & 0 & 1 & 0 \\
0 & 0 & 0 & 0 & 0 & 1
\end{array}\right),\\
G_{(235)} &= \left(\begin{array}{@{}R@{}R@{}R@{}R@{}R@{}R@{}}
0 & 0 & 1 & 0 & -1 & 0 \\
0 & 1 & 0 & 0 & 0 & 0 \\
0 & 0 & 1 & 0 & 0 & 0 \\
0 & 0 & 0 & 1 & 0 & 0 \\
0 & 0 & 0 & 0 & 1 & 0 \\
0 & 0 & 0 & 0 & 0 & 1
\end{array}\right),&\quad
G_{(241)} &= \left(\begin{array}{@{}R@{}R@{}R@{}R@{}R@{}R@{}}
1 & 0 & 0 & 0 & 0 & 0 \\
0 & 0 & 0 & 0 & 0 & 0 \\
0 & 0 & 1 & 0 & 0 & 0 \\
0 & 0 & 0 & 1 & 0 & 0 \\
0 & 0 & 0 & 0 & 1 & 0 \\
0 & 0 & 0 & 0 & 0 & 1
\end{array}\right),\\
G_{(243)} &= \left(\begin{array}{@{}R@{}R@{}R@{}R@{}R@{}R@{}}
1 & 0 & 0 & 0 & 0 & 0 \\
1 & 0 & 0 & 1 & 0 & 0 \\
0 & 0 & 1 & 0 & 0 & 0 \\
0 & 0 & 0 & 1 & 0 & 0 \\
0 & 0 & 0 & 0 & 1 & 0 \\
0 & 0 & 0 & 0 & 0 & 1
\end{array}\right),&\quad
G_{(245)} &= \left(\begin{array}{@{}R@{}R@{}R@{}R@{}R@{}R@{}}
1 & 0 & 0 & 0 & 0 & 0 \\
0 & 0 & 1 & 0 & 0 & -1 \\
0 & 0 & 1 & 0 & 0 & 0 \\
0 & 0 & 0 & 1 & 0 & 0 \\
0 & 0 & 0 & 0 & 1 & 0 \\
0 & 0 & 0 & 0 & 0 & 1
\end{array}\right),\\
G_{(251)} &= \left(\begin{array}{@{}R@{}R@{}R@{}R@{}R@{}R@{}}
1 & 0 & 0 & 0 & 0 & 0 \\
0 & 1 & 0 & 0 & 0 & 0 \\
0 & 0 & 0 & 0 & 0 & 0 \\
0 & 0 & 0 & 1 & 0 & 0 \\
0 & 0 & 0 & 0 & 1 & 0 \\
0 & 0 & 0 & 0 & 0 & 1
\end{array}\right),&\quad
G_{(253)} &= \left(\begin{array}{@{}R@{}R@{}R@{}R@{}R@{}R@{}}
1 & 0 & 0 & 0 & 0 & 0 \\
0 & 1 & 0 & 0 & 0 & 0 \\
1 & 0 & 0 & 0 & 1 & 0 \\
0 & 0 & 0 & 1 & 0 & 0 \\
0 & 0 & 0 & 0 & 1 & 0 \\
0 & 0 & 0 & 0 & 0 & 1
\end{array}\right),\\
G_{(254)} &= \left(\begin{array}{@{}R@{}R@{}R@{}R@{}R@{}R@{}}
1 & 0 & 0 & 0 & 0 & 0 \\
0 & 1 & 0 & 0 & 0 & 0 \\
0 & 1 & 0 & 0 & 0 & 1 \\
0 & 0 & 0 & 1 & 0 & 0 \\
0 & 0 & 0 & 0 & 1 & 0 \\
0 & 0 & 0 & 0 & 0 & 1
\end{array}\right),&\quad
G_{(341)} &= \left(\begin{array}{@{}R@{}R@{}R@{}R@{}R@{}R@{}}
1 & 0 & 0 & 0 & 0 & 0 \\
0 & 1 & 0 & 0 & 0 & 0 \\
0 & 0 & 1 & 0 & 0 & 0 \\
0 & 0 & 0 & 0 & 0 & 0 \\
0 & 0 & 0 & 0 & 1 & 0 \\
0 & 0 & 0 & 0 & 0 & 1
\end{array}\right),\\
G_{(342)} &= \left(\begin{array}{@{}R@{}R@{}R@{}R@{}R@{}R@{}}
1 & 0 & 0 & 0 & 0 & 0 \\
0 & 1 & 0 & 0 & 0 & 0 \\
0 & 0 & 1 & 0 & 0 & 0 \\
-1 & 1 & 0 & 0 & 0 & 0 \\
0 & 0 & 0 & 0 & 1 & 0 \\
0 & 0 & 0 & 0 & 0 & 1
\end{array}\right),&\quad
G_{(345)} &= \left(\begin{array}{@{}R@{}R@{}R@{}R@{}R@{}R@{}}
1 & 0 & 0 & 0 & 0 & 0 \\
0 & 1 & 0 & 0 & 0 & 0 \\
0 & 0 & 1 & 0 & 0 & 0 \\
0 & 0 & 0 & 0 & 1 & -1 \\
0 & 0 & 0 & 0 & 1 & 0 \\
0 & 0 & 0 & 0 & 0 & 1
\end{array}\right),\\
G_{(351)} &= \left(\begin{array}{@{}R@{}R@{}R@{}R@{}R@{}R@{}}
1 & 0 & 0 & 0 & 0 & 0 \\
0 & 1 & 0 & 0 & 0 & 0 \\
0 & 0 & 1 & 0 & 0 & 0 \\
0 & 0 & 0 & 1 & 0 & 0 \\
0 & 0 & 0 & 0 & 0 & 0 \\
0 & 0 & 0 & 0 & 0 & 1
\end{array}\right),&\quad
G_{(352)} &= \left(\begin{array}{@{}R@{}R@{}R@{}R@{}R@{}R@{}}
1 & 0 & 0 & 0 & 0 & 0 \\
0 & 1 & 0 & 0 & 0 & 0 \\
0 & 0 & 1 & 0 & 0 & 0 \\
0 & 0 & 0 & 1 & 0 & 0 \\
-1 & 0 & 1 & 0 & 0 & 0 \\
0 & 0 & 0 & 0 & 0 & 1
\end{array}\right),\\
G_{(354)} &= \left(\begin{array}{@{}R@{}R@{}R@{}R@{}R@{}R@{}}
1 & 0 & 0 & 0 & 0 & 0 \\
0 & 1 & 0 & 0 & 0 & 0 \\
0 & 0 & 1 & 0 & 0 & 0 \\
0 & 0 & 0 & 1 & 0 & 0 \\
0 & 0 & 0 & 1 & 0 & 1 \\
0 & 0 & 0 & 0 & 0 & 1
\end{array}\right),&\quad
G_{(451)} &= \left(\begin{array}{@{}R@{}R@{}R@{}R@{}R@{}R@{}}
1 & 0 & 0 & 0 & 0 & 0 \\
0 & 1 & 0 & 0 & 0 & 0 \\
0 & 0 & 1 & 0 & 0 & 0 \\
0 & 0 & 0 & 1 & 0 & 0 \\
0 & 0 & 0 & 0 & 1 & 0 \\
0 & 0 & 0 & 0 & 0 & 0
\end{array}\right),\\
G_{(452)} &= \left(\begin{array}{@{}R@{}R@{}R@{}R@{}R@{}R@{}}
1 & 0 & 0 & 0 & 0 & 0 \\
0 & 1 & 0 & 0 & 0 & 0 \\
0 & 0 & 1 & 0 & 0 & 0 \\
0 & 0 & 0 & 1 & 0 & 0 \\
0 & 0 & 0 & 0 & 1 & 0 \\
0 & -1 & 1 & 0 & 0 & 0
\end{array}\right),&\quad
G_{(453)} &= \left(\begin{array}{@{}R@{}R@{}R@{}R@{}R@{}R@{}}
1 & 0 & 0 & 0 & 0 & 0 \\
0 & 1 & 0 & 0 & 0 & 0 \\
0 & 0 & 1 & 0 & 0 & 0 \\
0 & 0 & 0 & 1 & 0 & 0 \\
0 & 0 & 0 & 0 & 1 & 0 \\
0 & 0 & 0 & -1 & 1 & 0
\end{array}\right).
\end{alignat*}
}%

The spectral radius of the matrix
{\renewcommand{\arraystretch}{0.7}%
\[
G_{(143)}G_{(231)}G_{(245)}G_{(342)}G_{(451)}G_{(124)}G_{(453)}
=
\left(\begin{array}{@{}R@{}R@{}R@{}R@{}R@{}R@{}}
0 & 0 & 0 & 0 & 0 & 0 \\
1 & -2 & 1 & 0 & 0 & 0 \\
0 & -1 & 1 & 0 & 0 & 0 \\
0 & 0 & 0 & 0 & 0 & 0 \\
0 & 0 & 0 & 0 & 1 & 0 \\
-1 & 1 & 0 & 0 & 0 & 0
\end{array}\right),
\]}

\noindent is $\frac{1+\sqrt{5}}{2}$, and its eigenvalues ordered by
decreasing of their modulus are:
\[
-\frac{1+\sqrt{5}}{2},~1,~\frac{\sqrt{5}-1}{2},~0,~0,~0.
\]

\end{document}